\documentclass[reqno,11pt]{amsart}
\usepackage[linktocpage]{hyperref}
\usepackage[utf8]{inputenc}
\usepackage{amsmath,amsthm,amsfonts,amssymb}
\usepackage{mathrsfs}
\usepackage{slashed}
\usepackage{dsfont}
\usepackage[mathscr]{eucal}

\renewcommand {\a}{ \alpha }
\renewcommand{\b}{\beta}

\newcommand{\g}{\gamma}

\renewcommand{\d}{\delta}
\newcommand{\s}{\sigma}

\renewcommand{\L}{\Lambda}
\newcommand{\z}{\zeta}
\renewcommand{\t}{\theta}

\newcommand{\p}{\partial}

\newcommand{\Om}{\Omega}

\newcommand{\Thanks}{\vspace*{.5em} \noindent \thanks}

\newcommand{\oq}{\ {\raise 7pt\hbox{${\scriptstyle\circ}$}}
	\kern -7pt{
		\hbox{$Q$}}}

\newcommand {\bb}{\mathbf b}

\newcommand {\BS}{\mathbf S}
\newcommand {\BR}{\mathbf R}

\newcommand {\BQ}{\mathbf Q}
\newcommand {\BT}{\mathbf T}

\newcommand {\bx}{\mathbf x}

\newcommand {\be}{\mathbf e}

\newcommand {\bk}{\mathbf k}

\newcommand {\bm}{\mathbf m}

\newcommand {\by}{\mathbf y}

\newcommand {\bu}{\mathbf u}
\newcommand {\bn}{\mathbf n}

\newcommand {\boldeta}{\boldsymbol\eta}

\newcommand {\bxi}{\boldsymbol\xi}

\newcommand{\lu}{\langle}
\newcommand{\ru}{\rangle}

\newcommand{\Fock}{{\mathcal{F}}}
\newcommand {\Four}{\mathscr{F}}

\newcommand{\CB}{\mathcal B}

\newcommand{\CH}{\mathcal H}

\newcommand{\CP}{\mathcal P}

\newcommand{\CA}{\mathcal A}

\newcommand{\CM}{\mathcal M}

\newcommand{\CC}{\mathcal C}

\newcommand{\CD}{\mathcal D}

\newcommand{\plainW}[1]{\textup{{\textsf{W}}}^{#1}}
\newcommand{\plainC}[1]{\textup{{\textsf{C}}}^{#1}}

\newcommand{\plainL}[1]{\textup{{\textsf{L}}}^{#1}}
\newcommand{\plainl}[1]{\textup{{\textsf{l}}}^{#1}}

\DeclareMathOperator*{\esssup}{ess-sup}

\DeclareMathOperator{\tr}{{tr}}
\DeclareMathOperator{\card}{{card}}
\DeclareMathOperator{\diag}{{\sf diag}}

\newcommand{\1}
{{\,\vrule depth3pt height9pt}{\vrule depth3pt height9pt}
	{\vrule depth3pt height9pt}{\vrule depth3pt height9pt}\,}

\DeclareMathOperator {\dist} {{dist}}

\DeclareMathOperator{\op}{{Op}}

\DeclareMathOperator{\dc}{d}

\newtheorem{thm}{Theorem}[section]
\newtheorem{cor}[thm]{Corollary}
\newtheorem{lem}[thm]{Lemma}
\newtheorem{prop}[thm]{Proposition}
\newtheorem{cond}[thm]{Condition}
\newtheorem{Remark}[thm]{Remark}

\theoremstyle{definition}
\newtheorem{defn}[thm]{Definition}

\newtheorem*{remark}{Remark}
\newtheorem{rem}[thm]{Remark}

\numberwithin{equation}{section}

\newcommand{\bee}{\begin{equation}}
	\newcommand{\ene}{\end{equation}}
\newcommand{\bees}{\begin{equation*}}
	\newcommand{\enes}{\end{equation*}}
\newcommand{\bes}{\begin{split}}
	\newcommand{\ens}{\end{split}}

\newcommand{\bet}{\begin{thm}}
	\newcommand{\ent}{\end{thm}}
\newcommand{\bel}{\begin{lem}}
	\newcommand{\enl}{\end{lem}}
\newcommand{\bec}{\begin{cor}}
	\newcommand{\enc}{\end{cor}}
\newcommand{\bep}{\begin{proof}}
	\newcommand{\enp}{\end{proof}}
\newcommand{\ber}{\begin{rem}}
	\newcommand{\enr}{\end{rem}}

\newcommand{\Z}{\mathbb Z}

\newcommand{\R}{ \mathbbm R}
\newcommand{\C}{ \mathbbm C}

\def\ceq{{\coloneqq}}

\newcommand{\scrM}{\mycal{M}}
\newcommand{\Pdd}{\mbox{$\partial$ \hspace{-1.2 em} $/$}}
\newcommand{\Sl}{\mbox{$\prec \!\!$ \nolinebreak}}
\newcommand{\Sr}{\mbox{\nolinebreak $\succ$}}

\newcommand{\Hmath}{\mathscr{H}}

\newcommand{\eps}{\varepsilon}

\newcommand{\QEDrem}{}

\DeclareFontFamily{OT1}{rsfso}{}
\DeclareFontShape{OT1}{rsfso}{m}{n}{ <-7> rsfso5 <7-10> rsfso7 <10-> rsfso10}{}
\DeclareMathAlphabet{\mycal}{OT1}{rsfso}{m}{n}

\makeatletter
\@namedef{subjclassname@2020}{%
	\textup{2020} Mathematics Subject Classification}
\makeatother

\usepackage{color,bbm,interval,mathtools}

\begin{document}
	\hoffset -4pc

\title[The Relativistic Fermionic Entanglement Entropy]{The fermionic entanglement entropy and area law
for the relativistic Dirac vacuum state}

\author[F.\ Finster]{Felix Finster}
\address{Fakult\"at f\"ur Mathematik  Universit\"at Regensburg  D-93040 Regensburg, Germany}
\email{finster@ur.de}

\author[M.\ Lottner]{Magdalena Lottner}
\address{Fakult\"at f\"ur Mathematik  Universit\"at Regensburg  D-93040 Regensburg, Germany}
\email{magdalena.lottner@ur.de}

\author[A.V.\ Sobolev]{Alexander V. Sobolev \\ October 2023 / June 2024}
\address{Department of Mathematics, University College London, Gower Street, London WC1E 6BT, United
Kingdom}
\email{a.sobolev@ucl.ac.uk}


\dedicatory{  }

\begin{abstract}
We consider the fermionic entanglement entropy for the free Dirac field in
a bounded spatial region of Minkowski spacetime.
In order to make the system ultraviolet finite, a regularization is introduced.
An area law is proven in the limiting cases where the volume tends to infinity
and/or the regularization length tends to zero.
The technical core of the paper is to generalize a theorem of Harold Widom
to pseudo-differential operators whose principal symbols develop a specific
discontinuity at a single point.
\end{abstract}

\maketitle 

\vspace*{-1.5em}

\tableofcontents

\section{Introduction}
Entropy is a measure for the disorder of a physical system.
There are various notions of entropy, like the entropy in classical statistical mechanics
as introduced by Boltzmann and Gibbs, the Shannon and R{\'e}nyi entropies in information theory
or the von Neumann entropy for quantum systems.
In the past decade, there has been increasing interest in the {\em{entanglement entropy}}, 
which tells about non-classical correlations between subsystems of a
composite quantum system~\cite{amico-fazio, horodecki}.
We here consider the {\em{fermionic}} case, where the many-particle system is composed of
fermions satisfying the (Pauli-)Fermi-Dirac statistics.
Moreover, for simplicity we consider the {\em{quasi-free}} case where the particles do not
interact with each other. This makes it possible to express the entanglement entropy
in terms of the reduced one-particle density operator~\cite{helling-leschke-spitzer}
(for details see Section~\ref{secentquasi}).
This setting has been studied extensively for a free Fermi gas formed of non-relativistic spin-less
particles~\cite{helling-leschke-spitzer, leschke-sobolev-spitzer, LSS_2022}
(for more details see the preliminaries in Section~\ref{secentquasi}).
In the present paper, we turn attention to a {\em{relativistic}} system formed of particles {\em{with spin}}.
More precisely, we consider a free Dirac field in a bounded spatial subset of Minkowski spacetime.
We compute the entanglement entropy for the quantum state describing the vacuum
with an ultraviolet regularization on a length
scale~$\eps$. The corresponding one-particle density operator turns out to be
the regularized projection operator to all negative-frequency solutions of the Dirac equation.
Making use of the explicit form of the Dirac propagator
and employing the techniques developed in the mathematical section of this paper,
it becomes possible to compute the entanglement entropy 
of bounded spatial subregions in Minkowski spacetime.
We prove an area law in the limiting cases where the volume tends to infinity
and/or the regularization length tends to zero.
The technical core of the paper is to generalize the pseudo-differential methods
for the one-particle density to principal symbols which develop a specific
discontinuity at a single point.

We now give an overview of our methods and results. 
Let~$\Pi$ be the projection onto the negative frequency subspace of the Dirac operator 
on~$\plainL2(\R^3;\mathbb C^4)$. 
As we recall in 
Section~\ref{secgreen}, in the momentum representation, this projection is simply 
the multiplication by 
the self-adjoint~$4\times 4$-matrix-valued function  
\begin{align}
	\label{def:cp}
	\mathcal P(\bk) :=\frac{1}{2}\: 
	\Bigg( \mathds{1}_{\C^4} + \frac{\displaystyle \sum \nolimits_{\beta=1}^{3} k_\beta \gamma^\beta \gamma^0 - m \gamma^0}{\sqrt{\bk^2+m^2}}  \Bigg)\:, \quad \bk\in \R^3.
\end{align}
In other words, the operator~$\Pi$ is a pseudo-differential operator on~$\plainL2(\R^3; \mathbb C^4)$ 
with \textit{symbol} $\CP(\bk)$,
\begin{equation} \label{Pidef}
(\Pi u)(\bx) = \big( \op(\CP) u \big)(\bx) 
= \frac{1}{(2\pi)^3} \iint e^{i\bk\cdot(\bx-\by)}\, \CP(\bk)\, u(\by)\,d\by \,d\bk \:,
\end{equation}
where the parameter~$m \geq 0$ is the rest mass of the Dirac particle.
Given a parameter~$\eps>0$ (the {\em{regularization length}}) and a function $\phi\in \plainC\infty(\R_+)$ with~$\phi(0)=1$ (the {\em{cutoff function}}),  in this paper we shall be mainly concerned 
with the following \textit{regularized} version of symbol~\eqref{def:cp}: 
\begin{align}
	\CP^{(\eps)}(\bk) := &\ \CP(\bk)\, 
	\phi\big(\eps\sqrt{\bk^2+m^2}\big)
	\notag \\
	= &\ 
	\frac{1}{2}\: \Bigg( \mathds{1}_{\C^4} + \frac{\displaystyle \sum \nolimits_{\beta=1}^{3} 
	k_\beta \gamma^\beta \gamma^0 - m \gamma^0}{\sqrt{\bk^2+m^2}}  \Bigg) 
	\:
\phi\Big(\eps\sqrt{\bk^2+m^2}\,\Big)
	\:, \quad \bk\in \R^3\:,\label{def:A_eps}
\end{align}
which gives rise to the \textit{regularized} projection~$\Pi^{(\eps)} = \op(\CP^{(\eps)})$. Further 
requirements on the function $\phi$ will be given in Theorem \ref{thm:mainPhys}. 

Now, for each~$\varkappa >0$ we introduce the \textit{R\'enyi entropy function}, which is defined 
as  follows. If~$t\notin (0,1)$ then we set~$\eta_\varkappa(t) = 0$. For~$t\in (0, 1)$ we define
\begin{align}
\begin{split}
	\label{eq:eta_gamma}
	\eta_\varkappa(t)= & \ 
			\displaystyle \frac{1}{1-\varkappa}\,\log \big( t^\varkappa + (1-t)^\varkappa \big)
			\qquad\;\;\;\: \text{for } \varkappa\neq 1\:,\\[0.2cm]
\eta_1(t):=\lim\limits_{\varkappa \rightarrow 1} \eta_\varkappa(t)
	= &\ -t \log t 
	- (1-t) \log (1-t) \qquad \text{for }  \varkappa = 1\;.
\end{split}
\end{align}
Note that~$\eta_1$ is the familiar von Neumann entropy function. 
Now, for any bounded (always assumed non-empty)
set~$\L\subset\R^3$  we can define the \textit{R\'enyi entanglement entropy} 
associated with the bi-partition~$\R^3 = \L \cup(\R\setminus\L)$ 
(see e.g.~\cite[Section~3]{leschke-sobolev-spitzer2}) by
\begin{align}
	\label{RenyEnt}
S_\varkappa(\Pi^{(\eps)}, \Lambda):=
\tr \big(\eta_\varkappa \big( \chi_{\Lambda} \: \Pi^{(\eps)}\: \chi_{\Lambda} \big) -\chi_{\Lambda} \,\eta_\varkappa( \Pi^{(\eps)})\,\chi_{\Lambda}\big) \:,
\end{align}
As we shall see later, 
since~$\L$ is bounded and~$\eps>0$, both operators on the right-hand side are trace class, so the 
entropy~$S_\varkappa$ is well-defined. 
Our main objective is to analyze the asymptotic behavior of the entropy 
$S_\varkappa(\Pi^{(\eps)}, L \Lambda)$ as the 
regularization parameter~$\eps$ tends to zero 
and/or the \textit{scaling parameter}~$L$ tends to infinity.  

The following theorem constitutes our main result --  it provides the area law 
for the asymptotics of the entanglement entropy. From now on, as a rule we assume that~$\L$ is a 
\textit{region}, i.e. an open set with finitely many connected components such that 
their closures are disjoint. 

\begin{thm}
	\label{thm:mainPhys} 
	Let~$\Lambda\subset \R^3$ be a bounded spatial region with~$\plainC1$-boundary. 
Let the cut-off $\phi\in\plainC\infty(\R_+)$ be a function such that 	
$\phi(0)=1$ and 
\begin{align}\label{eq:phi}
\left|\frac{d^k}{d t^k}\, \phi(t)\right|\le C_k \:(1+t)^{-\rho},\ t >0,
\end{align} 
with some constants $C_k$ for all $k = 0, 1, \dots$, where $\rho >3$.
	Then, as~$L\eps^{-1}\to \infty$ and~$\eps \searrow 0$, the following asymptotics hold:
	\begin{align}
		\label{eq:areaLaw}
		\lim \,L^{-2}\eps^2\: 
		S_\varkappa(\Pi^{(\eps)}, L \Lambda)  = \mathfrak{M}_\varkappa \,  
		\mathrm{vol}_2(\partial \Lambda) \:,
	\end{align}
	where~$\mathfrak{M}_\varkappa$ is some explicit constant. 
	
If~$L\to\infty$ and~$\eps>0$ is fixed, then 
	\begin{align}
		\label{eq:areaLawL}
		\lim  \,L^{-2} \eps^2\:  
		S_\varkappa(\Pi^{(\eps)}, L \Lambda)  = \mathfrak{M}_\varkappa^{(\eps)}\,
		\mathrm{vol}_2(\partial \Lambda) \:,
	\end{align}
where~${\mathfrak M}_\varkappa^{(\eps)}$ is some explicit constant such that 
$\mathfrak{M}_\varkappa^{(\eps)}\to \mathfrak{M}_\varkappa$ as~$\eps \searrow 0$.

If~$0 < \varkappa <2$, then both coefficients~$\mathfrak{M}_\varkappa$ 
and~$\mathfrak{M}_\varkappa^{(\eps)}$ 
are strictly positive. 
%
\end{thm} \noindent
The definitions of the 
coefficients~$\mathfrak M_\varkappa, \mathfrak M_\varkappa^{(\eps)}$ require more  
technical preliminaries and are given in the proof of Theorem~\ref{thm:mainPhys} in Section~\ref{Sec:Appl}.
We point out that these coefficients depend on the choice of the cutoff function~$\phi$.
Note that the above theorem also applies to the limiting case where~$L$ is fixed and~$\eps \searrow 0$.
 
We now comment on the {\em{role of the ultraviolet cutoff}}. We first point out that, in the setting of
relativistic quantum field theory, the ultraviolet regularization is needed in order for the entropy to be well-defined
and finite. This well-known fact can be understood in various ways. One way is to note that a physically sensible
state should satisfy a micro-local energy condition. This condition can be stated equivalently
that the state should be a {\em{Hadamard state}}~\cite{radzikowski}. In our setting of a quasi-free state,
this means that the two-point distribution~$P(x,y)$ has a singularity structure of Hadamard form.
Consequently, the corresponding reduced one-density operator~$\Pi$ defined by~\eqref{Pidef}
is not trace class (because its kernel~$\Pi(\bx,\by)$ is singular on the diagonal~$\bx=\by$).
Moreover, corresponding entanglement entropies diverge.
An alternative, more elementary way of understanding the divergence of the entropy is to note that
the description of anti-particles in quantum field theory makes it necessary to take
the negative-energy solutions of the Dirac equation into account. In fact, the operator~$\Pi$ in~\eqref{Pidef}
is the projection operator onto all negative-frequency solutions. The fact that the total local density of these
solutions is infinite at each spacetime point suggests that also the corresponding entanglement entropies
diverge.

There are various ways to avoid this divergence. One strategy is to restrict attention to the
{\em{relative entropy}} which, being defined as the difference of two traces, can be finite even if the individual
traces are not. The relative entropy has interesting connections to modular theory
(see for example~\cite{hollands-sanders, galanda2023relative, witten-entangle}).
Another strategy is to compute the entanglement entropy not between a spatial region and its complement
but instead the entanglement entropy between two separated regions in spacetime~\cite{hollands-islam}.
Here we take the point of view that, similar to the ultraviolet divergences in perturbative quantum field theory,
the divergence of the entropy shows that quantum field theory should not hold on all scales, but only
down to a microscopic length~$\varepsilon$ (typically thought of the Planck length)
where it is to be replaced by another, more fundamental physical theory.
We remark that the theory of causal fermion systems is a concise proposal for such a theory
(for details see for example the textbooks~\cite{intro, cfs} or the website~\cite{cfsweblink}).
Here we do not need to be specific on how this microscopic physical theory should look like.
Instead, we work with standard quantum field theory, but take into account its limitations by introducing
a regularization on the length scale~$\varepsilon$. Then the entropy becomes finite, and
the smallness of the Planck length motivates to consider the asymptotics~$\varepsilon \searrow 0$.

Clearly, this procedure raises the question how our results depend on the choice of the regularization.
In order to address this question, with the ansatz~\eqref{def:A_eps} we consider a
cutoff function~$\phi$ which can be chosen in a general 
class of radially symmetric and smooth functions.
The area laws in Theorem~\ref{thm:mainPhys} hold irrespectively 
of the choice of the cutoff function.
As we shall see later, only the numerical values of the 
coefficients~$\mathfrak{M}_\varkappa$ and~$\mathfrak{M}_\varkappa^{(\eps)}$
depend on this choice (see again the proof of Theorem~\ref{thm:mainPhys} in Section~\ref{Sec:Appl}).
Considering non-smooth cutoff functions could lead to an asymptotic behavior different from our area law,  which we want to avoid.  
For example, as explained in \cite[Theorem 1.1]{bollmann-mueller}, a sharp, step function cutoff at Fermi energy $\mu >0$ entails the 
\textit{enhanced area law} asymptotics (i.e.\ asymptotics with an extra logarithmic factor 
$\log L$).  
We should point out that in our paper, since $\Pi$ is the projection on the negative spectral subspace, apart from the smooth cutoff we also have a sharp cutoff 
at energy zero. This however does not lead to the enhanced area law since the associated singularity of the symbol is concentrated at the single point $\bk = 0$, whereas the enhanced area law is produced 
when the singularity is supported on a surface of co-dimension one. 
 
From the mathematical viewpoint, our analysis is based on the quasi-classical 
asymptotic theory for truncated pseudo-differential operators with matrix-valued symbols, 
see \cite{Widom1980, widombook}. We focus on pseudo-differential operators on 
$\plainL2(\R^d, \mathbb C^n)$, $d\ge 2, n\ge 1$, of the following form:
\begin{align}\label{eq:pdo}
\big( \op_\a(\CA) u\big) (\bx) 
: = \Big( \frac{\alpha}{2\pi} \Big)^d 
\,\iint \: e^{i \a \bxi (\bx - \by)} \CA(\bxi)\, u(\by)\: d \by \: d\bxi 
\qquad \text{for any } \bx \in \R^d\:,
\end{align}
where~$\CA$ is the $(n\times n)$-matrix-valued symbol, and~$\a >0$ is the scaling parameter. 
%
%
Introduce the truncated operator: 
\begin{align}\label{eq:truncpdo}
W_\a(\CA, \L) := \chi_\L\,\op_\a(\CA) \chi_\L,
\end{align}
where, as earlier, $\chi_\L$ denotes the multiplication by the indicator of a set~$\L\subset\R^d$.
We are interested in the asymptotics of the trace of 
the operator
\begin{align}\label{eq:dal}
D_\a(\CA, \L; f) := \chi_{\L} f\big(W_\a(\CA, \L)\big) \chi_\L - W_\a(f\circ\CA, \L), 
\end{align}
as~$\a\to\infty$.  The function~$f$, which we call \textit{test function}, is continuous, 
but is not supposed to be smooth.
Whenever it does not cause confusion we use the notation $D_\a(\CA)$ or $D_\a(\CA; f)$.

One should note that the asymptotic regime~$\a\to\infty$ can be interpreted in two ways. First, 
the reciprocal value~$\a^{-1}$ can be regarded as a quasi-classical parameter, i. e. \textit{Planck's constant}, in which case the sought asymptotics is simply the quasi-classical asymptotics. 
At the same time, a straightforward change of spatial variable 
$\bx\to \a^{-1}\bx$ shows that~$D_\a(\CA, \L; f)$ is unitarily equivalent to~$D_1(\CA, \a\L; f)$, and hence 
$\a\to\infty$ amounts to the spatial scaling limit. 
We should also observe that if~$f(0)=0$, $\L$ is bounded and~$\CA$ decays at infinity sufficiently 
fast, then each of the operators on the right-hand side of~\eqref{eq:dal} is trace class. 
In particular, the trace 
of the second operator equals 
\begin{align*}
\bigg(\frac{\a}{2\pi}\bigg)^{d}\,\mathrm{vol}_{d}(\L)\,\int_{\R^d} 
\tr f\big(\CA(\bxi)\big) \:d\bxi \:.
\end{align*}
If~$f(0)\not = 0$ or~$\L$ has infinite volume,
then the two operators 
are no longer trace class individually, but 
under some mild conditions their difference, i.e.\ the operator~$D_\a$, is. 

Using~$D_\a(\CA, \L; f)$ we can express the operator under the trace 
in the definition~\eqref{RenyEnt}. Indeed,  taking~$d = 3$ and 
changing the variable~$\bk = \bxi \eps^{-1}$ we can rewrite~$\Pi^{(\eps)}$ as 
\begin{align*} 
\Pi^{(\eps)} = &\ \op_{1/\eps}(\mathcal A^{(\eps)}) \quad \textup{with}\\
\mathcal A^{(\eps)}(\bxi) = &\ \CP^{(\eps)}(\bxi \eps^{-1})
= 
\frac{1}{2}\: \Bigg( \mathds{1}_{\C^4} + \frac{ \displaystyle \sum \nolimits_{\beta=1}^{3} \xi_\beta \gamma^\beta  - \eps m }
{\sqrt{\bxi^2+\eps^2m^2}} \gamma^0 \Bigg)
\:
\phi\big(\sqrt{\bxi^2+(\eps m)^2}\big)\:,
\end{align*}
and hence 
\begin{align*}
S_\varkappa (\Pi^{(\eps)}, \L) = \tr D_{1/\eps} \big( \CA^{(\eps)}, \L; \eta_\varkappa \big) \:.
\end{align*}
Furthermore, due to the equivalence with the scaling limit,  we can write 
\begin{align*}
S_\varkappa (\Pi^{(\eps)}, L \L) = \tr D_{\a} \big( \CA^{(\eps)}, \L; \eta_\varkappa \big)
\qquad \text{with} \qquad \a = L\eps^{-1}. 
\end{align*}
To study the trace we 
rely on the asymptotic 
formula 
%
%
obtained 
by H.\ Widom in \cite{Widom1980},
%
%
see Proposition~\ref{prop:widom}. There are two significant complications however. 
The first one is that Widom's result was proved for analytic test functions~$f$, whereas 
$\eta_\varkappa(t)$ is not even differentiable at~$t = 0$ or~$t=1$ if~$\varkappa\le 1$. 
The second issue is that Widom's original formula holds for smooth symbols~$\CA$, but 
the symbol~$\CA^{(\eps)}$ becomes discontinuous in the limit~$\eps \searrow 0$:
\begin{align}
	\label{eq:Alimit}
	\CA^{(\eps)}(\bxi)\to \CA(\bxi)=\frac{1}{2}\: \bigg( \mathds{1}_{\C^4} + \sum_{\beta=1}^{3} \frac{\xi_\beta}{|\bxi|}\gamma^\beta
	\gamma^0 \bigg)
\phi(|\bxi|)\:.
%
%
\end{align} 
Note that in the case~$m=0$, the symbol~$\CA^{(\eps)}$ 
coincides with the limit symbol~\eqref{eq:Alimit} for all~$\eps\ge 0$.  
 Thus our objective is to extend Widom's formula to non-smooth functions~$f$ 
 and non-smooth symbols~$\CA$. More precisely, we focus on non-smooth test functions and non-smooth 
 symbols that mimic the properties of~$\eta_\varkappa$ and the symbol~\eqref{eq:Alimit} 
 respectively. Note that we consider arbitrary dimensions~$d\ge 2$ and not only~$d=3$. Moreover, 
 instead of one discontinuity point (as in~\eqref{eq:Alimit}), we allow for an arbitrary finite number 
 of discontinuity points.  
 
 The asymptotic behavior of~$D_\a(\CA, \L; f)$ 
 for non-smooth~$f$ has been explored for smooth \textit{scalar} symbols~$\CA$  
in \cite{leschke-sobolev-spitzer2}, 
\cite{leschke-sobolev-spitzer}, \cite{LSS_2022} and 
\cite{sobolev-schatten}, \cite{sobolev-functions}, and some of these 
techniques are used in this paper. 
The mainstay of our work is an adaptation 
%
%
of the existing Schatten-von Neumann estimates for the operator~\eqref{eq:dal} 
to non-smooth (matrix) symbols. Subsequently, these estimates allow us to approximate 
non-smooth symbols 
by smooth ones and hence enable us to use the original Widom's formula. 
Secondly, they allow us to generalize Widom's formula to non-smooth functions following 
the ideas of \cite{leschke-sobolev-spitzer} and \cite{sobolev-functions}.  

%

The paper is structured as follows.
In Section~\ref{Sec:PhysPrel} we recall a few physical 
preliminaries on the Dirac equation in Minkowski spacetime and entanglement entropy. 
In Section~\ref{sect:widom} we provide some mathematical background of our analysis including the 
asymptotic formula by H. Widom, see Proposition~\ref{prop:widom}. Our main mathematical results are 
contained in Theorems~\ref{thm:main} and~\ref{thm:maineps}. 
Theorem~\ref{thm:main} obtains the asymptotics 
in one parameter, as~$\a\to\infty$, for a fixed symbol~$\CA$. 
The more general Theorem~\ref{thm:maineps} considers a convergent (as~$\eps \searrow 0$) family of non-smooth symbols~$\CA^{(\eps)}$ and establishes a Widom type asymptotic formula in two parameters, 
$\eps \searrow 0$ and~$\a\to\infty$.  
In Section~\ref{Sec:Schatten} we collect the key technical bounds on Schatten-von Neumann 
(quasi-)norms of pseudo-differential operators, including the ones with non-smooth symbols. 
Section~\ref{Sec:EstDalpha}  
focuses on the properties of the operator~$D_\a(\CA, \L; f)$ and prepares the grounds for 
approximations by smooth symbols. The conclusions of Section~\ref{Sec:EstDalpha} are used 
in Section~\ref{Sec:ProofSmooth} to prove Theorem~\ref{thm:main} for smooth symbols~$\CA$. 
The proofs of Theorems~\ref{thm:main}  and~\ref{thm:maineps} 
are completed in the short Section~\ref{Sec:ProofGeneral}.
In Section~\ref{sec:Positivity} we use the well-known result by Berezin from \cite{Berezin} for 
concave functions to examine the positivity of the asymptotic coefficient in the main theorems. 
And finally, 
in Section~\ref{Sec:Appl} 
the results of Sections~\ref{sect:widom} and~\ref{sec:Positivity} are applied to 
the entanglement entropy~$S_\varkappa(\Pi^{(\eps)}, L\L)$ to 
complete the proof of our main result, Theorem~\ref{thm:mainPhys}. \\

\noindent
{\bf{Units and notational conventions.}} 
We work throughout in natural units~$\hbar = c = 1$. Then the only remaining unit is that of
a length (measured for example in meters). It is most convenient to work with dimensionless
quantities. This can be achieved by choosing an arbitrary reference length~$\ell$ and multiplying
all dimensional quantities by suitable powers of~$\ell$. For example, we work with the
\begin{equation} \label{dimensionless}
\text{dimensionless quantities} \qquad
m \ell\:,\quad \bxi \ell\:,\quad \frac{\bx}{\ell} \quad \text{and} \quad \frac{\eps}{\ell} \:.
\end{equation}
For ease in notation, in what follows we set~$\ell=1$, making it possible to
leave out all powers of~$\ell$. The dimensionality can be recovered by rewriting
all formulas using the dimensionless quantities in~\eqref{dimensionless}.

We conclude the introduction with some general notational conventions.
For~$\bxi\in\R^d$ we denote~$\lu \bxi\ru  = \sqrt{1+|\bxi|^2}$. The
symbol~$\mathrm{vol}_{n}(\Om)$ with~$n = 0, 1, \dots, d$ stands the  
$n$-dimensional Lebesgue measure of the (measurable) set~$\Om\subset\R^d$.
 We call~$\Om\subset\R^d$ a {\em{region}} if it is a non-empty
 open set with finitely many connected components such that 
their closures are disjoint. 
We also use the standard notation $\R_+ = (0, \infty), \R_- = (-\infty, 0)$.

For any matrix~$\CB$ 
the notation~$|\CB|$ stands for its Hilbert-Schmidt norm. 
If~$\CB = \CB(\bxi), \bxi\in\R^d,$ is a smooth matrix-valued function then we write 
\begin{align*}
|\nabla_{\bxi}^l \CB(\bxi)| = \sum_{|\bm|=l}|\p_{\bxi}^\bm \CB(\bxi)|,  
  \end{align*}
where~$\p_{\bxi}^\bm~$ is the standard partial derivative of order~$\bm \in \mathbb Z_+^d$. 
The symbol~$\C^{n\times m}$ denotes the space of all ($n\times m$)-matrices. 

For any two bounded self-adjoint operators~$A$ and~$B$ on the Hilbert space~$\mathcal H$ 
the inequality~$A\le B$ is understood in the 
standard quadratic form sense:~$(Au, u)\le (Bu, u)$ for all 
$u\in \mathcal H$. 

For two non-negative numbers (or functions) 
$X$ and~$Y$ depending on some parameters, 
we write~$X\lesssim Y$ (or~$Y\gtrsim X$) if~$X\le C Y$ for
some positive constant~$C$ independent of those parameters.
 To avoid confusion we may comment on the nature of 
(implicit) constants in the bounds. 
If~$X\lesssim Y\lesssim X$, then we write~$X\asymp Y$.

\section{Physical preliminaries}
\label{Sec:PhysPrel}
\subsection{The Dirac equation in Minkowski spacetime} \label{secgreen}
In special relativity, space and time are combined to a
four-dimensional {\em{spacetime}}. Mathematically, this
four-dimensional spacetime is described by Minkowski spacetime
$(\scrM, \langle.,. \rangle)$, a real four-dimen\-sio\-nal vector space endowed
with an inner product~$\langle . , . \rangle$ of signature
$(+ \ \!\! - \ \!\! - \ \! - )$. For~$\scrM$ one may always choose a
basis~$(e_i)_{i=0,\ldots,3}$ satisfying~$\langle e_0,e_0 \rangle = 1$ and
$\langle e_i,e_i \rangle = -1$ for~$i=1,2,3$. Such a basis is called
pseudo-orthonormal basis or \emph{reference frame}, since the
corresponding coordinate system~$(x^i)$ describes time and space as
observed by an observer in a system of inertia. We also refer to
$t := x^0$ as time and denote spatial coordinates by
$\bx=(x^1, x^2, x^3)$. Representing two vectors
$x,y \in \scrM$ in such a basis as
$x = \sum_{i=0}^3 x^i e_i$ and~$y = \sum_{i=0}^3 y^i e_i$,
the inner product takes the form
\[ 
\langle x, y \rangle
= \sum_{j,k=0}^3 g_{jk}\: x^j\: y^k\:, \]
where~$g_{ij}$, the {\em{Minkowski metric}}, is the diagonal matrix
$g={\mbox{diag}}\,(1,-1,-1,-1)$.

The Dirac equation for a wave function~$\psi \in \plainC{\infty}(\scrM, \C^4)$ of mass~$m \geq 0$
in the Minkowski vacuum (i.e.\ without external potential) reads
\begin{flalign} \label{direq}
(i \Pdd - m)\, \psi(x) = 0 \:,
\end{flalign}
where we use the slash notation with the Feynman dagger~$\Pdd := \gamma^j \partial_j$
(for more details on the Dirac equation see~\cite{thaller} or~\cite[Sections~1.2, 1.3 ]{intro}).
We always work with the Dirac matrices in the Dirac representation
\[ \gamma^0 = \left( \begin{array}{cc} \mathds{1}_{\mathbb{C}^2} & 0 \\ 0 & -\mathds{1}_{\mathbb{C}^2} \end{array} \right) ,\qquad \vec{\gamma} = \left( \begin{array}{cc}
0 & \vec{\sigma} \\ -\vec{\sigma} & 0 \end{array} \right) \:, \]
and~$\vec{\sigma}$ are the three Pauli matrices
\[ 
    \sigma^1 =
    \begin{pmatrix}
        0 & 1
        \\ 1 & 0
    \end{pmatrix}
    , \qquad
    \sigma^2 =
    \begin{pmatrix}
        0 & -i \\
        i & 0
    \end{pmatrix}
    , \qquad
    \sigma^3 =
    \begin{pmatrix}
        1 & 0
        \\ 0 & -1
    \end{pmatrix} \:. \]
The wave functions at a space-time point~$x$ take values in the {\em{spinor space}}~$S_x$, a four-dimensional complex vector space endowed with an indefinite inner product of signature~$(2,2)$,
which we call {\em{spin inner product}} and denote by
\begin{equation} \label{sip}
\Sl \psi | \phi \Sr_x = \sum_{\alpha=1}^4 s_\alpha\: \psi^\alpha(x)^\dagger
\phi^\alpha(x) \:,\qquad s_1=s_2=1,\;\; s_3=s_4=-1\:,
\end{equation}
where~$\psi^\dagger$ is the complex conjugate wave function (in the physics literature, this inner product is
often written as~$\overline{\psi} \phi$ with the so-called adjoint spinor
$\overline{\psi} = \psi^\dagger \gamma^0$).
%
Since the Dirac equation is linear and hyperbolic
(meaning that it can be rewritten as a symmetric hyperbolic system; for details see
for example~\cite[Chapter~13]{intro}), 
its Cauchy problem for smooth initial data is well-posed,
giving rise to global smooth solutions. Moreover, due to finite propagation speed, starting with
compactly supported initial data, we obtain solutions which are spatially compact (meaning that
their restriction to any Cauchy surface has compact support).
For any two such solutions~$\psi, \phi$, the vector field~$\Sl \psi | \gamma^j \phi \Sr$ is divergence-free;
this is referred to as {\em{current conservation}}. Applying Gauss' divergence theorem,
this implies that the spatial integral
\begin{equation} \label{print}
(\psi | \phi) \big|_t :=  \int_{\R^3} \Sl \psi | \gamma^0 \phi \Sr|_{(t, \bx)}\: d^3\bx
\end{equation}
is independent of the choice of the space-like hypersurface labelled by the time parameter~$t$.
This integral defines a scalar product on the space of smooth solutions of the Dirac equation
with spatially compact support. Forming the completion, we obtain a Hilbert space, which we
denote by~$(\Hmath, (.|.))$. The norm on~$\Hmath$ is denoted by~$\| . \|$.

In what follows, we will work in a fixed coordinate system with a distinguished time function~$t=x^0$.
In this setting, it is most convenient to write the Dirac equation in an equivalent way which resembles
the Schr\"odinger equation. To this end, we multiply
the Dirac equation~\eqref{direq} by~$\gamma^0$ and isolate the~$t$-derivative on
one side of the equation,
\begin{equation}
\label{Hamilton}
i \partial_t  \psi = H \psi
\qquad \text{where} \qquad
H := -\gamma^0 (i
\vec{\gamma} \vec{\nabla} - m)
\end{equation}
(note here that~$\sum_{j=0}^3 \gamma^j \partial_j = \gamma^0 \partial_0 + \vec{\gamma}
\vec{\nabla}$). The operator~$H$ is referred to as the {\em{Dirac Hamiltonian}},
and~\eqref{Hamilton} is the Dirac equation in the {\em{Hamiltonian form}}.
By direct computation one verifies that the Hamiltonian is a symmetric operator on
the Hilbert space~$\Hmath$. Working at fixed time~$t=0$, in view of~\eqref{print},
the Hilbert space~$\Hmath$ can be identified with with the square-integrable spinors,
\[ \Hmath = L^2(\R^3, \C^4) \:. \]
In what follows, we shall always work with this identification.

Applying the (unitary) Fourier transform
\[ 
\big( \Four \phi \big) (\bk) := \frac{1}{\sqrt{(2 \pi)^3}} \int e^{-i\bk \bx} \phi(\bx) \: d\bx \:,
\quad \phi\in L^2(\R^3,\C^4) \:,
\] 	
the Hamiltonian~$H$ may be rewritten as
\[  
H= \Four^{-1} \Bigg(  \sum_{\beta=1}^{3} k_\beta \gamma^\beta +m \Bigg)\: \gamma^0\, \Four \:.  
\]
Using that the Hermitian $4 \times 4$-matrix~$\big(  \sum_{\beta=1}^{3} k_\beta \gamma^\beta +m \big)\: \gamma^0$
is trace-free and that its square can be computed with the help of the anti-commutation relations
to be~$\bk^2+m^2$ times the identity matrix, we conclude that that its eigenvalues
are~$\pm \sqrt{\bk^2+m^2}$, both with multiplicity two.
Hence, diagonalizing this matrix with a suitable unitary matrix~$S$ gives
\[ \bigg(  \sum_{\beta=1}^{3} k_\beta \gamma^\beta +m \bigg) \gamma^0 = S^{-1} J S 
\quad \text{with} \quad J := \sqrt{\bk^2 + m^2} \: \diag \big( -1,-1,1,1 \big) \:.  \]
Therefore, the projection onto the negative spectral subspace of~$H$ is given by
\begin{equation} \label{Pidef2}
\Pi := (S\,\Four)^{-1} \:\frac{1}{2} \bigg( \mathds{1}_{\C^4} - \frac{1}{\sqrt{\bk^2+m^2}}\: J \bigg)\: S\, \Four \:.
\end{equation}
Inserting the regularizing factor~%
%
%
$\phi\big(\eps\sqrt{\bk^2+m^2}\big)$, 
where~$\eps>0$ is the regularization length we obtain the \emph{regularized projection operator}
\begin{align}
\begin{split}
	\label{eq:KernelFermProj}
	\Pi^{(\eps)} :=& \: (S\Four)^{-1} \:\frac{1}{2}\: 
\phi\big(\eps\sqrt{\bk^2+m^2}\big)	
	%
	\: \Big( \mathds{1}_{\C^4} - \frac{1}{\sqrt{\bk^2+m^2}} \:J \Big)\: S\, \Four \\
	=&\:  \Four^{-1} \:\frac{1}{2} \:
\phi\big(\eps\sqrt{\bk^2+m^2}\big)	
%
	\: \bigg( \mathds{1}_{\C^4} - \frac{ \big( \sum_{\beta=1}^{3} k_\beta \gamma^\beta +m \big)\: \gamma^0 }{\sqrt{\bk^2+m^2}} \bigg)\: \Four\:.
\end{split}
\end{align}

\begin{Remark} {\bf{(Connection with the kernel of the fermionic projector)}} 
{\em{
For simplicity, we here restrict attention to the Hamiltonian formulation and work
exclusively with operators acting on the spatial Hilbert space~$L^2(\R^3, \C^4)$.
Nevertheless, the operator~$\Pi^{(\eps)}$ is closely related to kernels in spacetime, as we now explain.
The subspace of negative-energy solutions of the 
Dirac equation in Minkowski space can be described
by the {\em{regularized kernel of the fermionic projector}}~$P^{(\eps)}(x,y)$ defined by
\[ 
P^{(\eps)}(x,y) = \int \frac{1}{(2\pi)^4} \bigg( \sum_{i=0}^3 k_j \gamma^j + m \bigg) \,\delta
\big( \langle k,k \rangle -m^2 \big)\, \Theta(-k^0)\: \phi \big( \eps |k^0| \big)
\: e^{-i \langle k, x-y \rangle} \:d^4k \:, \]
where the cutoff function~$\phi$ and the parameter~$\eps>0$ again describes the regularization,
$\langle \cdot, \cdot \rangle$ is the Minkowski inner product, and~$\Theta$ denotes the Heaviside function.
This kernel plays a central role in the theory of causal fermion systems 
(for more details see \cite[Chapter~1]{cfs} or~\cite[Chapters~15 and~16]{intro}). In this context, the most common and simplest choice is an exponential cutoff
\[ \phi(\tau) = e^{- \tau}, \quad \tau \ge 0 \]
(for more details and some background see~\cite[Sections~1.2.2 and~2.4.1]{cfs}). We emphasize however that in the 
current paper we allow rather arbitrary cut-offs $\phi$ satisfying the decay properties  stated in Theorem 
\ref{thm:mainPhys}.

If we choose both arguments on the Cauchy surface~$\{t=0\}$, i.e.\
\[ x = (0, \bx),\quad y = (0, \by) \qquad \text{with~$\bx, \by \in \R^3$} \]
and carry out the integral over~$k^0$, we obtain
\[ P^{(\eps)} \big( (0,\bx),\, (0, \by) \big) 
=  \frac{1}{(2\pi)^4} \int \frac{1}{2 |\omega|}\:
(\slashed{k} + m) \, \phi( \eps |\omega|) \Big|_{\omega = -\sqrt{\bk^2+m^2}}\;
\: e^{i \bk (\bx-\by)} \:d^3\bk \:. \]
Comparing with~\eqref{eq:KernelFermProj}, one sees that
\[ 
\Pi^{(\eps)}(\bx,\by) = -2\pi\: P^{(\eps)}((0,\bx), (0, \by)\big)\:\gamma^0	 \:. \]
Hence the integral kernel of the spatial operator~$\Pi^{(\eps)}$ is obtained from the
regularized kernel of the fermionic projector simply by multiplying with a prefactor and with the
matrix~$\gamma^0$ from the right. This matrix~$\gamma^0$ will appear frequently in our formulas;
it can be understood as describing the transition from the setting in a Lorentzian spacetime
to the purely spatial formulation on a given Cauchy surface.
}} \QEDrem
\end{Remark}

\subsection{The entanglement entropy of a quasi-free fermionic quantum state} \label{secentquasi}
Given a Hilbert space~$(\Hmath, \langle .|. \rangle_\Hmath)$ (the ``one-particle Hilbert space''),
we let~$(\Fock, \langle .|. \rangle_\Fock)$ be the corresponding fermionic Fock space, i.e.\
\[ \Fock = \bigoplus_{k=0}^\infty \;\underbrace{\Hmath \wedge \cdots \wedge \Hmath}_{\text{$k$ factors}} \]
(where~$\wedge$ denotes the totally anti-symmetrized tensor product).
We define the {\em{creation operator}}~$\Psi^\dagger$ by
\[ \Psi^\dagger \::\: \Hmath \rightarrow \text{\rm{L}}(\Fock) \:,\qquad
\Psi^\dagger(\psi) \big( \psi_1 \wedge \cdots \wedge \psi_p \big) := \psi \wedge \psi_1 \wedge \cdots \wedge \psi_p \:, \]
(where~$\text{\rm{L}}(\Fock)$ denotes the linear mappings from~$\Fock$ to itself).
Its adjoint is the annihilation operator denoted by~$\Psi(\overline{\psi}) := (\Psi^\dagger(\psi))^*$.
These operators satisfy the canonical anti-commutation relations
\[ \label{CAR}
	\big\{ \Psi(\overline{\psi}), \Psi^\dagger(\phi) \big\} = (  \psi | \phi ) \qquad\text{and}\qquad
	\big\{ \Psi(\overline{\psi}), \Psi(\overline{\phi}) \big\} = 0 = \big\{ \Psi^\dagger(\psi), \Psi^\dagger(\phi) \big\} \:. \]
Next, we let~$W$ be a {\em{statistical operator}} on~$\Fock$, i.e.\ a positive semi-definite linear operator of trace one,
\[ W \::\: \Fock \rightarrow \Fock\:,\qquad W \geq 0 \quad \text{and} \quad \tr_\Fock(W)=1 \:. \]
Given an observable~$A$ (i.e.\ a symmetric operator on~$\Fock$), the expectation value of the measurement
is given by
\[ \langle A \rangle := \tr_\Fock\big( A W) \:. \]
The corresponding {\em{quantum state}}~$\omega$ is the linear functional which to every observable
associates the expectation value, i.e.\
\[ \omega \::\: A \mapsto \tr_\Fock\big( A W) \:. \]
The {\em{von Neumann entropy}} of the state~$\omega$ is defined by
\[ 
S := - \tr_\Fock \big( W \, \log W \big) \:. \]

In this paper, we restrict our attention to the subclass of so-called \emph{quasifree} states, fully determined by their two-point distributions
\[ \omega_2(\overline{\psi}, \phi) := \omega\big( \Psi^\dagger(\phi)\, \Psi(\overline{\psi}) \big) \:. \]
\begin{defn}
The {\bf{reduced one-particle density operator}}~$D$ is the positive linear operator on the Hilbert
space~$(\Hmath, \langle .|. \rangle_\Hmath)$ defined by
\[ \omega_2(\overline{\psi}, \phi) = \langle \psi \,|\, D \phi\rangle_\Hmath \:. \]
\end{defn} \noindent
The von Neumann entropy~$S_1(\omega)$ of the quasi-free fermionic
quantum state can be expressed in terms of the reduced one-particle density operator by
\begin{equation} \label{Sred}
S_1(\omega) = \tr \eta(D) \:,
\end{equation}
where~$\eta = \eta_1$ is the von Neumann entropy function defined in~\eqref{eq:eta_gamma}. 
This formula appears commonly in the literature
(see for example~\cite[Equation 6.3]{ohya-petz}, \cite{klich, casini-huerta, longo-xu}
and~\cite[eq.~(34)]{helling-leschke-spitzer}).
A detailed derivation can be found in~\cite[Appendix~A]{fermientropy}.
Similar to~\eqref{Sred} also other entropies can be expressed in terms of the reduced one-particle density
operator. Namely, the R{\'e}nyi entropy and the corresponding entanglement entropy can be written
as~$S_\varkappa(\omega) = \tr \eta_\varkappa(D)$ and~\eqref{RenyEnt}, respectively.
These formulas are also derived in~\cite[Appendix~A]{fermientropy}.

We here consider a quasi-free state formed of solutions of the Dirac equation in Minkowski space.
Thus we choose the Hilbert space~$\Hmath$ as the solution space of the Dirac equation
with scalar product~$\langle .|. \rangle = (.|.)$.
Moreover, we consider the {\em{regularized vacuum state}}, in which case the reduced two-particle density
operator is equal to the regularized projection operator onto all negative-frequency solutions
of the Dirac equation, i.e.\
\[ D=\Pi^{(\eps)} \]
with~$\Pi^{(\eps)}$ as in~\eqref{eq:KernelFermProj}.
We point out that, in the limiting case~$\eps \searrow 0$, the operator~$D$ goes over to the
projection operator to all negative-frequency solutions~\eqref{Pidef2}.
The corresponding quantum state~$\omega$ is the vacuum state in Minkowski space.

\begin{Remark} {\em{
	We point out that our definition of entanglement entropy differs from
	the conventions in~\cite{helling-leschke-spitzer, leschke-sobolev-spitzer2}
	in that we do not add the entropic difference operator of the complement of~$\Lambda$.
	This is justified as follows. On the technical level, our procedure is easier, because
	it suffices to consider compact spatial regions.
	Conceptually, restricting attention to the entropic difference operator of~$\Lambda$
	can be understood from the fact that occupied one-particle states which are supported either inside or outside~$\Lambda$
	do not contribute to the entanglement entropy.
	Thus it suffices to consider the one-particle states which are non-zero both inside and outside.
	These ``boundary states'' are taken into account already in our definition of
	the entanglement entropy~\eqref{RenyEnt}. }} \QEDrem
\end{Remark}

\section{Widom's formula and its generalizations}\label{sect:widom}
As explained in the Introduction, our approach is based on the asymptotic analysis 
of truncated pseudo-differential operators on  
$\plainL2(\R^d; \mathbb C^n)$, $d\ge 2, n\ge 1$, with 
matrix-valued symbols~$\CA(\bxi)$, that are  defined by~\eqref{eq:pdo}, \eqref{eq:truncpdo}. 
In the scalar case, i.e. for~$n=1$, we often use for symbols the lower case letters and 
write~$a(\bxi)$ instead of~$\CA(\bxi)$. 
In this and the next few sections the objective is to study the asymptotics of the trace of 
the operator~$D_\a(\CA, \L; f)$ defined in~\eqref{eq:dal}, 
as~$\a\to\infty$, for non-smooth symbols~$\CA$ and suitable test functions~$f$.  
Sometimes, one or more arguments will be omitted if it does not cause confusion, and we simply
write~$D_\a(\CA; f)$ or~$D_\a(\CA)$.
In order to state our results, we need to describe first  the asymptotic coefficient. 
For a vector~$\be\in \mathbbm S^{d-1}$, we represent~$\bxi\in\R^d$ as  
\begin{align*}
	\bxi = \hat\bxi + t \be \qquad \text{with } t\in \R \text{ and }\hat\bxi \in 
	\BT_\be := \{\bxi \: | \: \bxi\cdot \be = 0 \} \:.
\end{align*}
Instead of~$\CA(\bxi)$ we sometimes write 
\begin{align*}
	\CA(\hat\bxi; t) := \CA(\hat\bxi + t \be) \:.
\end{align*}
Fixing the value $\hat\bxi$ and considering $\CA(\hat\bxi, t)$ as a function 
of one real variable $t$, we introduce the operator 
$W_1(\CA(\hat\bxi; \,\cdot\,); \R_+)\big)$ acting on $\plainL2(\R; \mathbb C^n)$.
For a function~$f:\R\to\mathbbm C$ denote
\begin{align*} 
	\CM(\hat\bxi; \be; \CA; f) 
	:= &\ \tr\big[\chi_{\R_+}\,f\big(W_1(\CA(\hat\bxi; \,\cdot\,); \R_+)\big)\, \chi_{\R_+}
	- W_1\big(f(\CA(\hat\bxi; \,\cdot\,)); \R_+\big)\big] \\
	= &\ \tr D_1\big( \CA(\hat\bxi; \,\cdot\,); \R_+; f\big),\quad \hat\bxi\in \BT_\be\:,
\end{align*} 
where $\R_+ = (0, \infty)$, and introduce  
\begin{align}\label{eq:frakm}
\mathfrak M(\be; \CA; f) := \frac{1}{(2\pi)^{d-1}}  
\, \int_{\BT_{\be}} \,\CM(\hat\bxi; \be; \CA; f) \, d\hat\bxi \:. 
\end{align}
Finally, denoting the outer unit normal to~$\partial \Lambda$ at~$\bx \in \partial \Lambda$ by~$\bn_\bx$,
we can define the main asymptotic coefficient: 
\begin{align}\label{eq:bcoef}
{\sf B}(\CA; f) : = 
\int_{\p\L} \mathfrak M(\bn_\bx; \CA; f)\, dS_\bx\:.
\end{align}
As Lemma~\ref{lem:mbound} later on shows, under the conditions of the main Theorems  
\ref{thm:main} and \ref{thm:maineps} the coefficients \eqref{eq:frakm} and \eqref{eq:bcoef} are well-defined. In particular, they are finite for $\plainC\infty$ symbols $\CA$ and 
polynomial $f$.
Our analysis stems from the following result for smooth symbols~$\CA$ due to H. Widom \cite{Widom1980}.    
Let~$\CA\in\plainC\infty(\R^d; \C^{n\times n})$, $d\ge 2$, be a matrix-valued symbol 
(not necessarily Hermitian) such that 
\begin{align}\label{eq:wid}
|\nabla_{\bxi}^k\CA(\bxi)|\lesssim \lu \bxi\ru^{-\rho} \qquad \text{with} \qquad k=0, 1, \dots \text{ and } \rho >d. 
\end{align}     
The following is Widom's result stated in the form adapted for our needs. 

\begin{prop}\label{prop:widom} 
Let~$d\ge 2$, and let~$\L\subset\R^d$ be a bounded $\plainC1$-region.  
Suppose that~$\CA\in \plainC\infty(\R^d; \C^{n\times n})$ is a matrix-valued symbol satisfying~\eqref{eq:wid}.
Let~$f$ be a polynomial with~$f(0) = 0$. Then 
\begin{align}\label{eq:widom}
\tr f(W_\a(\CA; \L)) = \bigg(\frac{\a}{2\pi}\bigg)^{d}\,\mathrm{vol}_{d}(\L)\,\int_{\R^d} 
\tr f\big(\CA(\bxi)\big) \:d\bxi + \a^{d-1} {\sf B}(\CA; f) + o(\a^{d-1}) \:.
\end{align}
\end{prop}

\begin{remark}
\begin{enumerate}
\item Under the conditions of Proposition~\ref{prop:widom}  
 both operators in definition~\eqref{eq:dal}  of~$D_\a(\CA, \L; f)$ are trace class, and 
 the first term on the right-hand side of~\eqref{eq:widom} is exactly the 
trace~$\tr\, W_\a(f\circ\CA, \L)$.  In this case we also have the equality 
$\chi_{\L} f\big(W_\a(\CA, \L)\big) \chi_\L = f\big(W_\a(\CA, \L)\big)$, and therefore the formula
\eqref{eq:widom} can be rewritten as
\begin{align}\label{eq:widom11}
\lim_{\a\to\infty}\,\a^{1-d} \tr D_\a(\CA, \L; f) = 
{\sf B}(\CA; f) \:.
\end{align}
On the other hand, if~$f$ is a polynomial such that~$f(0)\not = 0$, then~\eqref{eq:widom} does not make sense, but~\eqref{eq:widom11} still holds.
\item
One should mention that in contrast with the matrix case, for scalar symbols 
$\CA = a$ the coefficient~${\sf B}(a; f)$ can be found explicitly, 
see \cite{widomtrace} and \cite{widombook}.  
\end{enumerate}
\end{remark}

The main mathematical contribution of our work is the extension of the above result to non-smooth 
symbols~$\CA$ and non-smooth functions~$f$. Remembering 
the relation~$\Pi^{(\eps)} = \op_1(\CP^{(\eps)})$, we 
study symbols that model the symbol~$\CP^{(\eps)}$ and its limit as~$\eps \searrow 0$. Namely, 
we consider symbols~$\CA$ that are~$\plainC\infty$ outside of a fixed finite 
set~$\boldsymbol\Xi = \{\bxi^{(1)}, \bxi^{(2)},\dots, \bxi^{(N)}\}\subset \R^d$, and satisfy the bound
\begin{align}\label{eq:Asing}
	|\nabla_{\bxi}^k \,\CA(\bxi)|\lesssim \lu \bxi\ru^{-\rho} 
	\dc(\bxi)^{-k},\quad \dc(\bxi) := \min\big(\dist(\bxi, \boldsymbol\Xi), 1\big),\quad 
	k = 0, 1, \dots \:,
\end{align}
where~$\rho > d$ and~$\lu \bxi\ru := \sqrt{1 + |\bxi|^2}$. 
We are also interested in families of symbols that converge in the following sense.
\begin{cond}\label{cond:Aeps} Let~$\boldsymbol{\Xi} \subset \R^d$ be a finite set.
	We assume that the family of Hermitian symbols 
	$\CA^{(\eps)}\in \plainC\infty(\R^d\setminus\boldsymbol\Xi; \C^{n\times n})$, $\eps\in [0, 1]$,   
	satisfies the bounds 
	\begin{align*}
		|\nabla_{\bxi}^k\, \CA^{(\eps)}(\bxi)|\lesssim \lu \bxi\ru^{-\rho} \dc(\bxi)^{-k},
		 \quad k = 0, 1, \dots,
	\end{align*}
	for some~$\rho > d$, uniformly in~$\eps$. 
	Away from~$\boldsymbol\Xi$ the symbols~$\CA^{(\eps)}$ converge to~$\CA:=\CA^{(0)}$ uniformly, i.e. 
	for each~$h>0$ we have
	\begin{align*}
		\sup_{\bxi: \dc(\bxi)>h}|\CA^{(\eps)}(\bxi) - \CA(\bxi)|\to 0,\quad \textup{as~$\eps \searrow 0$}\:.
	\end{align*}  
\end{cond}
In our analysis we can include functions~$f$ that are more general than the 
entropy functions~\eqref{eq:eta_gamma}:

\begin{cond}
	\label{cond:f2}
	Assume that there is a finite collection of points 
	${\sf T} = \{t_1, t_2, \dots, t_K\}\subset\R$ and a number~$\g\in (0, 1]$ such that the
	function~$f\in \plainC{}(\R)\cap \plainC2(\R\setminus\sf T)$ satisfies the bounds 
\[ 
\bigg|\frac{d^k}{dt^k} f(t) \bigg|\lesssim \sum_{l=1}^K |t-t_l|^{\g-k} 
\quad \text{for}\  k =1, 2, \quad \text{and}\ t\notin \sf T\:. \]
\end{cond}
 
Using these conditions we can now state the main technical results of this paper. 
The first of them deals with fixed non-smooth symbols~$\CA$. 

\begin{thm}
	\label{thm:main}
	Let~$d\ge 2$, and let 
$\L\subset\R^d$ be a bounded~$\plainC1$-region.
	Assume that the function~$f$ satisfies Condition~\ref{cond:f2} for some 
	$\g\in (0, 1]$. 
Suppose that a Hermitian matrix-valued symbol 
$\CA\in\plainC\infty(\R^d\setminus\boldsymbol\Xi; \C^{n\times n})$	
	 satisfies~\eqref{eq:Asing} for
	a finite set~$\boldsymbol{\Xi} \subset \R^d$ and~$\rho>0$ with~$\rho\g > d$. 
	Then the formula
	\begin{align}\label{eq:widom1}
\lim_{\a\to\infty}\,\a^{1-d} \tr D_\a(\CA, \L; f) = {\sf B}(\CA; f)
\end{align}	
	holds.
\end{thm}

The next theorem considers families of convergent symbols.

\begin{thm}
	\label{thm:maineps}
	Let~$d\ge 2$, and let the region~$\L\subset\R^d$  and the function~$f$ be as in Theorem~\ref{thm:main}.  
	Suppose that the family of Hermitian matrix-valued symbols 
	$\CA^{(\eps)}$ satisfies Condition~\ref{cond:Aeps} for some~$\rho>0$ with~$\rho\g > d$. Then, as~$\a\to\infty$ and~$\eps \searrow 0$, 
	we have  
	\begin{align}\label{eq:double}
		\lim\, \a^{1-d}\, \tr\, D_\a(\CA^{(\eps)}, \L; f) = {\sf B}(\CA; f) \:.
	\end{align}
\end{thm}

As the next Proposition shows, if the symbol 
$\CA$ is ``radially symmetric", then the integral 
\eqref{eq:frakm} is independent of the unit vector~$\be$, which simplifies the expression 
for the coefficient~${\sf B}(\CA; f)$.  

\begin{prop}\label{prop:area} 
Let~$d\ge 2$. 
Suppose that~$f$ satisfies Condition~\ref{cond:f2} with some~$\g\in (0, 1]$,  
and that~$\p\L$ is a union of finitely many bounded 
piece-wise~$\plainC1$-surfaces. 
Suppose that a Hermitian matrix-valued symbol 
$\CA\in\plainC\infty(\R^d\setminus\boldsymbol\Xi; \C^{n\times n})$ 
satisfies~\eqref{eq:Asing} with some~$\rho >0$ such that~$\rho\g > d$.  
Suppose also that for each~$\BR\in \mathrm{SO}(d)$ there exists a matrix 
$\BQ = \BQ_{\mathbf R}\in \mathrm{SU}(n)$  such that 
\begin{align}\label{eq:inva}
\CA( \bxi) = \BQ \,\CA(\BR \bxi)\, \BQ^* \qquad \textup{for a.e.~$\bxi\in\R^d$}\:.
\end{align}
Then~$\boldsymbol\Xi = \{0\}$, and 
the integral~\eqref{eq:frakm} does not depend on the vector~$\be\in\mathbb S^{d-1}$ and 
\begin{align}\label{eq:split}
{\sf B}(\CA; f) = \mathfrak{M} (\CA; f)\, \mathrm{vol}_{d-1}(\partial \Lambda) \:,
\end{align}
where we have denoted
$\mathfrak{M} (\CA; f) = \mathfrak M(\be; \CA; f)$ 
for an arbitrary~$\be$. 
\end{prop}

The identity~$\boldsymbol\Xi = \{0\}$ is an immediate consequence of the symmetry~\eqref{eq:inva}. 
Indeed if~$\boldsymbol\Xi$ contained a singular 
point~$\bxi_0\not = 0$, then by~\eqref{eq:inva} the symbol $\CA$ would have 
a singularity on the sphere~$ |\bxi| = |\bxi_0|$, which is not a finite set. 

Note that in Proposition~\ref{prop:area}, the region~$\L$ is not supposed to be bounded. In fact, any 
$\L$ satisfying the following condition, would be suitable:

\begin{cond}\label{cond:domain} 
The set~$\L\subset \R^d, d\ge 2$, is a region with 
piece-wise~$\plainC1$-smooth boundary,  
and  either~$\L$ or~$\R^d\setminus\L$ is bounded. 
\end{cond}

We note that~$\L$ and~$\R^d\setminus\overline{\L}$ 
satisfy Condition~\ref{cond:domain} simultaneously. 
The boundedness of~$\L$ in Theorems~\ref{thm:main} and~\ref{thm:maineps} 
is assumed only because both of them are derived from 
Proposition~\ref{prop:widom} where~$\L$ is supposed to be bounded. 
Most of the intermediate results in the forthcoming sections hold 
for the regions satisfying 
Condition~\ref{cond:domain}. 

Throughout the paper the implicit constants 
in the bounds involving the region $\L$, will of course depend on the region $\L$, but we 
will not point this out every time.

\section{Schatten-von-Neumann bounds for pseudo-differential operators} 
\label{Sec:Schatten}
\subsection{Singular values and Schatten-von Neumann classes} 
In this section we state some definitions and results on singular values and Schatten-von Neumann classes.
We refer to \cite[Ch.~11]{birman-solomjak} for more details on this topic. 

For a compact operator~$A$ in a separable Hilbert space~$\CH$ we denote 
by~$s_k(A), k = 1, 2, \dots,$ its singular values i.e. eigenvalues 
of the self-adjoint compact operator~$ \sqrt{A^*A}$ labeled in non-increasing order counting 
multiplicities.
For the sum~$A+B$ the following inequality holds: 
\begin{align}\label{eq:kyfan}
s_{2k}(A+B)\le s_{2k-1}(A+B)\le s_k(A) + s_k(B) \:. 
\end{align}
We say that~$A$ belongs to the Schatten-von Neumann class~$\BS_p$, $p>0$, 
if 
\[ \|A\|_p \coloneqq \big(\tr (A^*A)^{\frac{p}{2}}\big)^{\frac{1}{p}} \]
is finite. The functional 
$\|A\|_p$ defines a norm if~$p\ge 1$ 
and a quasi-norm if~$0<p<1$. With this (quasi-)norm the class~$\BS_p$ is a complete space. 
For~$0<p \leq 1$ the quasi-norm is actually a \emph{$p$-norm}, that is,  
it satisfies the following ``triangle inequality" for all~$A, B\in\BS_p$:
\begin{align}\label{eq:ptriangl}
\|A+B\|_p^p\le \|A\|_p^p + \|B\|_p^p\:.
\end{align} 
This inequality is used systematically in what follows. We point out a useful estimate 
for individual eigenvalues for operators in~$\BS_p$:
\begin{align}\label{eq:ind}
s_k(A)\le k^{-\frac{1}{p}}\, \|A\|_p,\ k = 1, 2, \dots.
\end{align}
Moreover, note that for block operators 
$A$ with entries~$A_{jk}$, $1\le j, k \le n$, where~$A_{jk}\in \BS_p$, $0 < p \le 1$, by 
\eqref{eq:ptriangl} we have   
\begin{align}
\label{est:SingValuesBlockOp}
\|A\|_{p}^p \leq \sum_{i,j=1}^n \|A_{ij}\|_p^p  \:, \qquad 0 < p \leq 1 \:,
\end{align} 
i.e. in order to obtain~$\BS_p$-bounds for~$A$ it suffices to obtain them for its entries.

\subsection{Lattice norm bounds}
Here we obtain bounds for suitable~$\BS_q$ (quasi-)norms for the operators 
$h \op_1(a)$, $h = h(\bx)$, $a = a(\bxi)$,
which have been studied quite extensively.

To this end we first introduce the following notation.
Let~$\mathcal C = [0, 1)^d\subset \R^d$ and let 
$\mathcal C_\bu = \mathcal C + \bu$ with~$\bu\in\R^d$.
For a function~$h\in\plainL{r}_{\textup{\tiny loc}}(\R^d)$, $r \in (0, \infty),$ 
denote
\begin{equation} \label{eq:brackh}
\begin{cases}
\1 h\1_{r,\delta} =
\biggr[\sum_{\bn\in\Z^d}
\biggl(\,\int\limits_{\CC_{\bn}} |h(\bx)|^r d\bx\biggr)^{\frac{\d}{r}}\biggr]^{\frac{1}{\d}} \:, &
0<\d<\infty\:,\\[0.3cm]
\1 h\1_{r, \infty} =
\sup_{\bu\in\R^d}
\biggl(\,\int\limits_{\CC_{\bu}} |h(\bx)|^r d\bx\biggr)^{\frac{1}{r}} \:, &
\d = \infty\:.
\end{cases}
\end{equation}
These functionals are sometimes called \textit{lattice  
quasi-norms} (norms for~$r, \d\ge 1$).  
If~$\1 h\1_{r, \d}<\infty$ we say that  
$h\in \plainl{\d}(\plainL{r})(\R^d)$.

We need the following
estimate from  \cite[Theorem 11.1]{birman-solomjak2} (see also \cite{birman-karadzov}, Section 5.8),
and quoted  in \cite[Theorem 4.5]{simon2005} for~$q\in [1, 2]$.

\begin{prop}\label{BS:prop}
Suppose that~$f\in  \plainl{q}(\plainL{2})(\R^d)$ 
and
$g\in  \plainl{q}(\plainL{2})(\R^d)$, 
with some~$q \in (0, 2]$. 
Let~$K: \plainL2(\R^d) \to \plainL2(\R^d)$ be the operator with the kernel 
\[
f(\bx) e^{i\bx\cdot \by}g(\by),\quad \bx\in\R^d, \by\in \R^d \:. 
\]
Then  
\[
\| K\|_{q} \lesssim \1 f\1_{2, q} \1 g\1_{2, q}
\]
with a constant independent of~$f$ and~$g$. 
\end{prop}

\subsection{Multi-scale symbols}
Consider~$a\in\plainC\infty(\R^d)$  
for which there exist positive continuous functions  
$v$ and~$\tau$ such that~$v$ is bounded and 
\begin{equation}\label{scales:eq}
|\nabla_{\bxi}^k a(\bxi)|\lesssim  
\tau(\bxi)^{-k} v(\bxi)\,,\ \qquad k = 0, 1, \dots, \quad \bxi\in\R^d\,,
\end{equation}
with constants independent of~$\bxi$. 
It is natural to call~$\tau$ the \emph{scale (function)}. 
The scale~$\tau$ is assumed to be globally Lipschitz 
with Lipschitz constant~$\nu <1$, that is, 
\begin{equation}\label{Lip:eq}
|\tau(\bxi) - \tau(\boldeta)| \le \nu |\bxi-\boldeta|\,,\ \ \mbox{ for all }\bxi,\boldeta\in\R^d\,.
\end{equation}
The function~$v$ is assumed to satisfy the bounds
\begin{equation}\label{w:eq}
v(\boldeta)\asymp v(\bxi)\,,\ \mbox{ for all }\boldeta\in B\bigl(\bxi, \tau(\bxi)\bigr),
\end{equation}
with implicit constants independent of~$\bxi$ and~$\boldeta$.  
For example, for smooth scalar symbols~$a$ the bound~\eqref{eq:wid} 
translates into~\eqref{scales:eq} with 
$v(\bxi) = \lu\bxi\ru^{-\b}$ and~$\tau(\bxi) = 1$.  In general, 
it is useful to think of~$v$ and~$\tau$ as 
(functional) parameters. They, in turn, can depend on other 
parameters, e.g.~numerical parameters like~$r\ge0$  in~\eqref{eq:Asing1} 
in the next section. 

We need only one result involving multi-scale symbols.

\begin{prop}\label{prop:cross_smooth} \cite[Lemma 3.4]{leschke-sobolev-spitzer}
Suppose that the region~$\L$ satisfies Condition~\ref{cond:domain}, and let 
the functions~$\tau$ and~$v$ be as described above. 
Suppose that the symbol~$a$ satisfies 
\eqref{scales:eq}, and that 
the conditions 	
\begin{align}\label{tau_low:eq}
\a \tau_{\textup{\tiny inf}}\gtrsim 1\,, \quad 
\tau_{\textup{\tiny inf}} \coloneqq \inf_{\bxi\in\R^d}\tau(\bxi)>0\,,
\end{align}
hold. Then  for any~$\s\in (0, 1]$ we have 
\[ 
\| [\op_\a(a), \chi_\L]\|_{\s}^\s
\lesssim \a^{d-1} \int_{\R^d} \frac{v(\bxi)^\s}{\tau(\bxi)}\,\mathrm{d}\bxi\:. \]
%
%
This bound is uniform in the symbols~$a$ satisfying~\eqref{scales:eq} with the same 
implicit constants. 
\end{prop}

\subsection{Extension to non-smooth symbols} 
For the symbols satisfying~\eqref{eq:Asing} 
one is tempted to apply Proposition~\ref{prop:cross_smooth} with 
$\tau(\bxi) = \dc(\bxi)$. However, this scale function vanishes at the points of the set 
$\boldsymbol\Xi$ and hence does not satisfy 
\eqref{tau_low:eq}. Nevertheless, as we see below, 
Proposition~\ref{prop:cross_smooth} is still applicable. 
In order to state the precise result, for methodological purposes instead of the condition 
\eqref{eq:Asing} we impose a more general condition on~$\CA$. Assume that for some finite set 
$\boldsymbol\Xi = \{\bxi^{(1)}, \bxi^{(2)},\dots, \bxi^{(N)}\}\subset \R^d$ and some 
$r\ge 0$ 
the Hermitian symbol~
$\CA\in \plainC\infty(\R^d\setminus\boldsymbol\Xi; \C^{n\times n})$ satisfies the bounds
\begin{align}\label{eq:Asing1}
|\nabla_{\bxi}^k\, \CA(\bxi)|\lesssim \lu |\bxi|+r\ru^{-\rho}\,  \dc(\bxi)^{-k},
\quad k = 0, 1, \dots, 
\end{align}
where, as in~\eqref{eq:Asing}, $\rho>d$ and~$\dc(\bxi) := \min\big(\dist(\bxi, \boldsymbol\Xi), 1\big)$. 

In what follows we need a partition of unity associated with the set~$\boldsymbol\Xi$. 
For~$h\in (0, 1/4]$ introduce the spatial scale function 
$\tau_h(\bxi) = \frac{1}{2}\big(\dc(\bxi)+h\big)$. 
It is clearly Lip\-schitz with Lip\-schitz constant~$1/2$. Therefore, 
using \cite[Theorem 1.4.10]{hormanderI} we can now construct two functions 
$\z_h^{(1)}, \z_h^{(2)}\in \plainC\infty(\R^d)$ with the properties 
\begin{align*}
\z_h^{(1)} + \z_h^{(2)}=1,\quad 
\z_h^{(1)}(\bxi) = 1\ \textup{for}\ \dc(\bxi) >  4h;  
\quad \z_h^{(2)}(\bxi) = 1\ \textup{for}\ \dc(\bxi) < 2h, 
\end{align*}
and such that 
\begin{align}\label{eq:cutoff}
|\nabla_{\bxi}^k\,\z_h^{(1)}(\bxi)|
+ |\nabla_{\bxi}^k\,\z_h^{(2)}(\bxi)|\lesssim (\dc(\bxi)+h)^{-k},\quad k= 0, 1, \dots, 
\end{align}
where the implicit constants depend only on dimension~$d$ and exponent~$k$. 

\begin{rem}\label{rem:uniform} Most of the bounds in this 
section will be uniform in 
the symbol $\CA$ and the point set $\boldsymbol\Xi$ in the following sense. 
We say that a bound is uniform in $\CA$ if it is uniform 
 in all symbols $\CA$ satisfying \eqref{eq:Asing} or \eqref{eq:Asing1} 
with the same implicit constants. Similarly, 
uniformity in $\boldsymbol\Xi$ means that the result is uniform in all point sets $\boldsymbol\Xi$ 
such that their cardinality $N=\card\boldsymbol\Xi$ is bounded above by the same constant. 

The uniformity in $\boldsymbol\Xi$ will be especially 
important in the proof of Lemma \ref{lem:mbound} further on. 
\end{rem}

\begin{lem} \label{lem:nonsm}
Suppose that the region~$\L$ satisfies Condition~\ref{cond:domain} and the Hermitian symbol~$\CA$ satisfies~\eqref{eq:Asing1} with some~$\rho >d$. 
Let~$\s\in (0, 1]$ be a number such that~$\rho \s > d$.   
Then for all~$\a \gtrsim 1$ we have 
\begin{align}\label{eq:nonsm}
\begin{cases}
	\| [\op_\a(\CA), \chi_\L]\|_{\s}^\s \lesssim \, \a^{d-1}\lu r\ru^{-\rho\s+d}
	& \textup{if~$d\ge 2$}\:,\\[0.2cm]
	\| [\op_\a(\CA), \chi_\L]\|_{\s}^\s \lesssim \,
	\lu r\ru^{-\rho\s+1}\log\big(\a +2\big)
	&\textup{if~$d=1$}\:.
\end{cases}
\end{align} 
The constants in~\eqref{eq:nonsm}  
are uniform in~$\CA$ satisfying~\eqref{eq:Asing1}, and in $\boldsymbol\Xi$, 
in the sense of Remark \ref{rem:uniform}.
\end{lem}

\begin{proof}
%
In view of~\eqref{est:SingValuesBlockOp} it is clear that it suffices to prove the bound~\eqref{eq:nonsm} for scalar symbols only. 
As in Proposition~\ref{prop:cross_smooth} we use for such symbols the notation~$a = a(\bxi)$. 
Using the partition of unity introduced before the lemma 
we split~$a(\bxi)$ into two symbols depending on~$h$:
\begin{align*}
a = &\ a_h^{(1)}+ a_h^{(2)},\\
a_h^{(1)} = a \z_h^{(1)},&\quad a_h^{(2)} = a \z_h^{(2)}.
\end{align*}
Let us consider first the symbol~$a_h^{(1)}$. Due to~\eqref{eq:Asing1} and~\eqref{eq:cutoff}, 
$a_h^{(1)}$ satisfies the bound
\begin{align*}
|\nabla^k a_h^{(1)}(\bxi)|\lesssim \lu |\bxi| + r\ru^{-\rho} (\dc(\bxi) + h)^{-k}.  
\end{align*}
Therefore, the function~$a_h^{(1)}$ satisfies the bound~\eqref{scales:eq} with  
$v(\bxi) = \lu |\bxi|+r\ru^{-\rho}$ and the scale function 
$\tau_h(\bxi) = \frac{1}{2}(\dc(\bxi) + h)$. It is straightforward 
that for these functions the conditions~\eqref{Lip:eq} and~\eqref{w:eq} are fulfilled
uniformly in $h\in (0, 1/4]$. Moreover, 
fixing ~$h = \a^{-1}$, we guarantee that~$\tau_h$ satisfies~\eqref{tau_low:eq}: 
\begin{align*}
\a \tau_h(\bxi)\ge \frac{1}{2}\a h = \frac{1}{2},\quad \bxi\in\R^d \:. 
\end{align*}
Thus Proposition~\ref{prop:cross_smooth} is applicable, and it gives the bound
\begin{align} \label{eq:ah1}
\| [\op_\a(a_{h}^{(1)}), \chi_\L]\|_{_\s}^\s \lesssim 
&\ 
\a^{d-1}\,
\int\limits_{\R^d} v(\bxi)^\s (\tau_h(\bxi))^{-1}\, d\bxi\notag\\
\lesssim &\ \a^{d-1}\, \lu r\ru^{-\rho\s}
\int\limits_{\dist(\bxi, \boldsymbol\Xi)\le 1} (\dc(\bxi) +  h)^{-1}\, d\bxi \nonumber\\
 &\qquad + \a^{d-1}\, \int\limits_{\dist(\bxi, \boldsymbol\Xi)> 1} \lu |\bxi| + r\ru^{-\rho \s}\, d\bxi\:.
\end{align}
Since
\begin{align*}
(\dist(\bxi, \boldsymbol\Xi) + h)^{-1}\le \sum_{j=1}^N (h+|\bxi - \bxi^{(j)}|)^{-1}
\end{align*}
in the first integral, the right-hand side of~\eqref{eq:ah1} does not exceed 
$\a^{d-1}\lu r\ru^{-\rho\s+d} $ if~$d\ge 2$ and 
\begin{align*}
	\lu r\ru^{-\rho\s} \log ( \a +2 ) + 
	\lu r\ru^{-\rho\s+1},
	 \quad \textup{if}\quad  d = 1\:.
\end{align*}
To estimate the commutator~$[\op_\a(a_h^{(2)}), \chi_\Lambda]$ we assume 
without loss of generality that~$\Lambda$ is bounded (if not, consider the commutator 
with~$\mathds{1}-\chi_\L$), and estimate separately 
its components~$\op_\a(a_h^{(2)}) \chi_\Lambda$  and 
$\chi_\Lambda \op_\a(a_h^{(2)})$. By the definition of the function~$\z_h^{(2)}$,  
its support can be covered by balls of radius~$4h$ centered at the points~$\bxi^{(j)}\in\boldsymbol\Xi$. 
Denote by~$\phi(\bxi)$ the indicator 
of the ball~$|\bxi|\le 4$, and observe that 
\begin{align*}
|a_h^{(2)}(\bxi)|\lesssim \sum_{j=1}^N b^{(j)}_h(\bxi),\quad b_h^{(j)}(\bxi) = 
\lu r \ru^{-\rho} \phi\big((\bxi-\bxi^{(j)})h^{-1}\big)\:.
\end{align*}
Thus it suffices to estimate each pair~$\op_\a(b_h^{(j)}) \chi_\Lambda$  and 
$\chi_\Lambda \op_\a(b_h^{(j)})$ individually. Fix the index~$j$ and 
assume without loss of generality that~$\bxi^{(j)} = 0$.
Then, by rescaling~$\bxi\to h \bxi$, the operator~$\op_\a(b_h^{(j)})$ 
coincides with~$\lu r \ru^{-\rho}\,\op_1(\phi)$ and therefore 
\begin{align*}
\|\op_\a(b_h^{(j)}) \chi_\Lambda\|_{\s}^\s + \|\chi_\Lambda \op_\a(b_h^{(j)})\|_{\s}^\s
\lesssim \lu r\ru^{-\rho \s} \| \chi_\L \op_1(\phi)\|_{\s}^\s \:. 
\end{align*}
As both functions~$\chi_\L$ and~$\phi$ are compactly supported, 
their lattice quasi-norms $\1\chi_\L\1_{2, \s}$ and $\1 \phi\1_{2, \s}$ 
(see \eqref{eq:brackh})  
are finite 
for any $\s\in (0, 1]$. Thus by Proposition~\ref{BS:prop}, the~$\BS_\s$-quasi-norm is bounded by a constant 
that depends only on dimension~$d$ and region~$\L$. 
As a consequence,   
\begin{align*}
	\|[\op_\a(a_h^{(2)}), \chi_\Lambda]\|_{\s}^\s\lesssim \lu r\ru^{-\rho\s} \:.
\end{align*}
Together with the previously derived bound for~$a_h^{(1)}$ this leads to 
\eqref{eq:nonsm}.
\end{proof}

\subsection{More estimates for the case~$d = 1$} 
We also need bounds for 
pseudo-differential operators in~$\plainL2(\R)$ given by 
\begin{align*}
\big(\op_\a^{\rm a}(p) u\big)(x)
= \frac{\a}{2\pi} \int\, e^{i\a(x-y)\xi} p(x, y, \xi) u(y)\, dy\, d\xi \:,
\end{align*}
where the scalar function~$p = p(x, y, \xi)$ is called \textsl{the amplitude}. Let 
\[ 
P_l(x, y, \xi; p) := \sum_{l_1, l_2 = 0}^l |\p_x^{l_1} \p_y^{l_2} p(x, y, \xi)| \:. 
\]
The following proposition is adapted from \cite[Theorem 2.5]{sobolev-schatten} and is stated in the form convenient for our purposes.  

\begin{prop}\label{prop:full} 
Let~$h_1, h_2~\in~\plainL{\infty}(\R)$, and let~$p$
be such that~$P_l\in \plainl{\s}(\plainL{1})(\R^{3})$
with some~$\s\in (0, 1]$ and ~$l = \lfloor\s^{-1}\rfloor+1$.
Then  
\begin{equation}\label{eq:amplitude}
	\|h_1 \op^{\rm a}_1(p) h_2\|_{\s} 
	\lesssim \| h_1\|_{\plainL\infty}
	\| h_2\|_{\plainL\infty} \1 P_l(p)\1_{1, \s} \:,
\end{equation}
with a constant independent of~$p$ and~$h_1, h_2$.
\end{prop}
A similar result holds also for operators in~$\plainL2(\R^d)$ with arbitrary~$d\ge 1$ 
but we do not need 
it here.

At this stage, for the one-dimensional estimates we need an assumption 
somewhat different from \eqref{eq:Asing1}. For~$r \ge 0$ and~$r_1 >0$ we assume that 
\begin{align}\label{eq:Asing11}
	|\p_{\xi}^k\, \CA(\xi)|\lesssim \lu |\xi|+r\ru^{-\rho}\, 
	(\dc(\xi) + r_1)^{-k}, \quad 
	k = 0, 1, \dots,
\end{align}
where $\rho >1$ and $\dc(\xi) = \min\{1, \dist(\xi, \boldsymbol\Xi)\}$ for some finite 
set $\boldsymbol\Xi\subset\R$.
Let~$\chi_\pm$ be the indicator function of~$\R_\pm$.

\begin{lem}\label{lem:dis}
Let~$d=1$. Suppose that~$\CA$ satisfies~\eqref{eq:Asing11}.
Let~$\s\in (0, 1]$ be some number such that~$\rho \s >1$.  
Then for all~$r >0$  and all~$r_1>0$ we have 
\begin{equation}\label{eq:onedim}
	\|\chi_\pm\op_1(\CA) \chi_{\mp}\|_{\s}^\s 
	\lesssim  
	\lu r\ru^{-\rho\s + 1}\, \log\bigg(\frac{1}{ r_1}+2\bigg),
\end{equation}
where the implicit constant is uniform in~$\CA$ and in the set $\boldsymbol\Xi$. 
\end{lem}

\begin{proof}  
Again, it suffices to study scalar symbols~$a$. 
We first estimate the quasi-norm of  
$ \chi_{(0, \a)}\op_1(a) \chi_-$, where~$\a\gtrsim 1$ will be chosen later. 
This operator is easily checked to be unitarily equivalent to 
\begin{align*}
	\chi_{(0, 1)}\op_\a(a) \chi_- = [\chi_{(0, 1)}, \op_\a(a)]\, \chi_-,
\end{align*}
and hence 
it satisfies the bound~\eqref{eq:nonsm}:
\begin{align}\label{eq:int}
	\|\,\chi_{(0, \a)}\op_1(a) \chi_-\|_{\s}^\s 
	\le &\ \|\,[\chi_{(0, 1)},\op_\a(a)]\|_{\s}^\s\notag\\
	\lesssim &\ \lu r\ru^{-\rho\s+1}\log (\a + 2) \:.
\end{align}
To estimate~$\chi_{(\a, \infty)} \op_1(a)\chi_-$ we use Proposition~\ref{prop:full}.  
Let~$\z\in\plainC\infty(\R)$ be a function such that~$\z(x) = 0$ for~$x \le 1/2$ 
and~$\z(x) = 1$ for 
$x\ge 1$. Then we can write
\begin{align*}
	\chi_{(\alpha,\infty)} \op_1(a) \chi_- = \chi_{(\alpha,\infty)} \op_1^{\rm a}(g) \chi_-,\quad g(x, y, \xi) = \z((x - y)/\a) \, a(\xi) \:.
\end{align*}
Let 
\begin{align*}
	L_u = \ -i u^{-1} \p_\xi   ,\quad u\not = 0. 
\end{align*}
Since~$L_u e^{iu\xi} = e^{iu\xi}$, we can integrate by parts as follows:
\begin{align*}
	\int e^{i(x-y)\xi} g(\xi)\,d\xi 
	= &\ i^k \int  e^{i(x-y)\xi} g^{(k)}(x, y, \xi)\, d\xi, \\
	\quad g^{(k)}(x, y, \xi) = &\ \z((x-y)/\a)(x-y)^{-k} \, \p_\xi^k a(\xi) \:.
\end{align*}
Therefore, for any~$k = 0, 1, \dots,$ we have~$\chi_{(\alpha,\infty)} \op_1(a) \chi_- = \chi_{(\alpha,\infty)} \op_1^{\rm a}(g^{(k)}) \chi_-$.
According to~\eqref{eq:amplitude}, the~$\BS_\s$-quasi-norm of 
$\chi_{(\alpha,\infty)} \op_1(a) \chi_-$ then estimates as 
\begin{align}\label{eq:fromprop}
	\|\chi_{(\a, \infty)}\,\op_1(a) \chi_-\|_{\s}^\s \lesssim \1 P_l(g^{(k)})\1_{1, \s} \:,
\end{align}
for~$l = \lfloor\s^{-1}\rfloor+1$.
Let us estimate the right-hand side remembering that~$|x|+|y|=|x-y|\ge \a$ for~$x \in (\alpha,\infty)$ and~$y \in \R^-$:
\begin{align*}
	P_l(x, y, \xi; g^{(k)})\lesssim \frac{|\p_\xi^k a(\xi)|}{|x|^k + |y|^k +\a^k} \:, 
\end{align*}
with an implicit  constant depending on $l$.
By~\eqref{eq:Asing11}, 
\begin{align*}
	P_l (x, y, \xi; g^{(k)})
	\lesssim (|x|^k + |y|^k +\a^k)^{-1} \lu |\xi|+r\ru^{-\rho}
	\,(\dc(\xi)+r_1)^{-k}.
\end{align*}
%
In order to estimate the~$(1, \s)$-quasi-norm write 
\begin{align*}
	\1 P_l (g^{(k)})\1_{1, \s}^\s
	\lesssim \1(|x|^k + |y|^k +\a^k)^{-1}\1_{1, \s}^\s\ 
	\1 \lu |\xi|+r\ru^{-\rho}  (\dc(\xi)+r_1)^{-k} \1_{1, \s}^\s \:, 
\end{align*}
where the first quasi-norm is taken in 
$\plainl{\s}(\plainL{1})(\R^{2})$, and the second one in~$\plainl{\s}(\plainL{1})(\R)$. 
Consider the factors individually assuming that 
$k\s >2$:
\begin{align*}
	\1 (|x|^k + |y|^k +\a^k)^{-1} \1_{1,\s}^\s
	\lesssim  \sum_{\bm\in\Z^2}(\a^k+|\bm|^k)^{-\s} \lesssim \a^{-k\s +2} \:.
\end{align*}
Now, 
\begin{align*}
	\1 \lu|\xi|+r\ru^{-\rho} (\dc(\xi)+r_1)^{-k} \1_{1, \s}^\s
	\lesssim  
r_1^{-k\s} 	
	\sum_{m\in\Z} \lu |m|+r\ru^{-\rho \s} 
	\lesssim  \lu r\ru^{-\rho \s+1}  r_1^{-k\s} \:.
\end{align*}
Therefore 
\begin{align*}
	\1 P_l(g^{(k)}) \1_{1, \s}^\s \lesssim \a^{-k\s +2}\,\lu r\ru^{-\rho \s+1} r_1^{-k\s} \:.
\end{align*}
Together with~\eqref{eq:fromprop} and~\eqref{eq:int} this yields the bound 
\begin{align}\label{eq:alpha}
	\|\chi_+\,\op_1(a) \chi_-\|_{\s}^\s \lesssim 
	\lu r\ru^{-\rho\s+1}\log (\a + 2) +  
	\a^{-k\s +2}\,\lu r\ru^{-\rho \s+1}   r_1^{ -k\s} \:.
\end{align}
Take 
$\a = \max (r_1^{-s}, 1)$, where~$s = k \s (k\s-2)^{-1}>0$, 
so that the right-hand side of~\eqref{eq:alpha} does not exceed 
$\lu r\ru^{-\rho\s+1}\log (r_1^{-1} + 2)$. 
Now~\eqref{eq:alpha} leads to~\eqref{eq:onedim}.
\end{proof}

\section{Estimates for the operator~$D_\a(\CA, \L; f)$ and
approximation by smooth symbols}\label{sect:model}
\label{Sec:EstDalpha}
\subsection{Two abstract results}

We begin with two results for compact operators in an arbitrary separable Hilbert space~$\CH$. 
Let~$A$ be a bounded operator and let~$P$ be an orthogonal projection on~$\CH$. For a function~$f$ 
define the operator 
\begin{align}\label{eq:cd}
\CD(A, P; f) := P f(PAP) P - P f(A) P \:.
\end{align}
In the next few statements, instead Condition~\ref{cond:f2} it will be more convenient to assume 
the following condition. 

\begin{cond}\label{cond:f}
The function~$f\in\plainC2(\R\setminus\{ t_0 \})\cap\plainC{}(\R)$ satisfies the bound 
\begin{align}\label{eq:fnorm}
\1 f\1_2 := \max_{0\le k\le 2}\sup_{t\not = t_0} |f^{(k)}(t)| |t-t_0|^{-\g+k}<\infty
\end{align}
for some~$\g\in (0, 1]$, and it is supported on the interval~$(t_0-R, t_0+R)$ with some finite 
$R>0$.  
\end{cond}

Condition~\ref{cond:f} is clearly less general than Condition~\ref{cond:f2}, but it allows to 
control the dependence of the quantities with which we work, on the size of the support of~$f$. 

The next proposition follows from a more general fact proved in 
\cite{sobolev-functions}, see also \cite[Proposition 2.2]{leschke-sobolev-spitzer}.

\begin{prop}\label{prop:szego} 
Suppose that~$f$ satisfies Condition~\ref{cond:f}  
with some~$\g\in (0, 1]$ and 
some~$t_0\in\R$, $R>0$. Let~$q \in (1/2, 1]$ and assume that   
$\s< \min(2-q^{-1}, \gamma)$.  
Let~$A$ be a bounded self-adjoint operator and let~$P$ be an orthogonal projection 
such that~$PA(\mathds{1}-P)\in \BS_{\s q}$. Then 	
\[ 
	\| \CD(A, P; f) \|_{q}
	\lesssim \1 f\1_2\, R^{\g - \s} \, \|P A (\mathds{1}-P)\|_{\s q}^\s \:, \]
with a positive implicit constant 
independent of the operators~$A$ and~$P$, the function~$f$, and the parameters~$R, t_0$.  
\end{prop}  

Next, we give a result comparing the operators~$\CD(A, P; f)$ and~$\CD(B, P; f)$. For this result 
we do not need 
the explicit dependence of the coefficients on the function~$f$.

\begin{prop} \protect{\cite[Corollary 4.3]{LSS_2022}} \label{prop:onfull} 
Let the function~$f$ satisfy Condition~\ref{cond:f} with some~$\g\in (0, 1]$, and let~$\s < \g$. 
Let~$A, B$ be bounded self-adjoint operators and let~$J$ be a bounded operator 
such that 
\[ 
	[A, J] =  [B, J] = 0\,,  \quad (A - B)J = 0\:. \]
Then  
\begin{align}\label{eq:onfull}
	\| \CD(A, P; f) - \CD(B, P; f)\|_{1} 
	\lesssim  &\ \|[(\mathds{1}-J)A, P]\|_{\s}^\s\notag\\ 
	&\ + \|[(\mathds{1}-J)B, P]\|_{\s}^\s + \|[J, P]\|_{\s}^\s + \|[J, P]\|_{1}\:.
\end{align}
The implicit constants in the above bound depend on~$\1 f\1_2$ and 
on the norms~$\|A\|, \|B\|$, and~$\|J\|$, 
and hence they are uniform in the operators~$A, B, J$ whose norms are bounded 
by the same constants.   
\end{prop} 

In the next subsection we apply the above propositions to pseudo-differential operators.  

\subsection{Estimates for~$D_\a(\CA, \L; f)$}
We apply this proposition to the operator~$A = \op_\a(\CA)$ and  projection~$P = \chi_\L$. 
As before, we suppose that~$\CA = \CA(\bxi)$ is a 
Hermitian matrix-valued symbol such that~\eqref{eq:Asing} holds. The bounds in this section 
are uniform in $\CA$ and $\boldsymbol\Xi$ in the sense of Remark \ref{rem:uniform}.

\begin{lem}\label{lem:dalpha}
Suppose that~$f$ satisfies Condition~\ref{cond:f} with some~$\g\in (0, 1]$,  and~$\L$ satisfies Condition~\ref{cond:domain}.
Let~$q\in (1/2, 1]$ and~$0<\s< \min(2-q^{-1}, \gamma)$.  
Assume that~$\rho\s q >d$. 
If~$d\ge 2$, then
\begin{align*}
	\|D_\a(\CA, \L; f)\|_{q}^q \lesssim \1 f\1_2^q\, R^{q(\g - \s)} \a^{d-1}\:,  
\end{align*}
and if~$d = 1$, then
\begin{align*}
	\|D_\a(\CA, \L; f)\|_{q}^q \lesssim \1 f\1_2^q\, R^{q(\g - \s)} \log (\a+2) \:,
\end{align*}
uniformly in $\CA$ and $\boldsymbol \Xi$.
%
\end{lem} 

\begin{proof}
By Proposition~\ref{prop:szego},
\begin{align*}
	\|D_\a(\CA, \L; f)\|_{q} \lesssim  &\ \1 f\1_2\, R^{\g - \s}\,
	\|\chi_\L\, \op_\a(\CA) (\mathds{1}-\chi_\L)\|_{\s q}^{\s}\\
	\le &\ \1 f\1_2\, R^{\g - \s}\, 
	\|[\op_\a(\CA), \chi_\L]\|_{\s q}^{\s} \:.
\end{align*}
The required bounds follow now from Lemma~\ref{lem:nonsm}. 
\end{proof}

\begin{lem}\label{lem:mbound}
Let~$d\ge 2$. 
Suppose that~$f$ satisfies Condition~\ref{cond:f} with some~$\g\in (0, 1]$,  and~$\L$ satisfies Condition~\ref{cond:domain}.
Suppose that~$\CA$ satisfies~\eqref{eq:Asing} with some~$\rho >0$ such that~$\rho\g > d$. 
Then for any~$\s \in (d \rho^{-1}, \g)$ we have 
\begin{align}\label{eq:mbound}
	|\CM(\hat\bxi; \be; \CA; f)|
	\lesssim \1 f\1_2\, R^{\g-\s}\, 
&\ \lu \hat\bxi\ru^{-\rho\s+1} \, \log \big(r(\hat\bxi)^{-1} + 2\big),\\	
	&\ \quad 
	\text{for~$\hat\bxi\in \BT_\be,\ \hat\bxi\notin \widehat{\boldsymbol\Xi}_\be$}\:,\notag
\end{align}
where~$\widehat{\boldsymbol\Xi} = \widehat{\boldsymbol\Xi}_\be$ denotes the projection of the
set~$\boldsymbol\Xi$ onto the hyperplane~
$\BT_\be = \{\bxi: \bxi\cdot \be = 0\}$, and~ 
$r(\hat\bxi) = \min\big(\dist(\hat\bxi, \widehat{\boldsymbol\Xi}_\be), 1\big)$. 
The bound is uniform in~$\CA$, $\boldsymbol\Xi$  and~$\be\in\mathbbm S^{d-1}$.

Furthermore, 
\begin{align}\label{eq:mathfr}
|\mathfrak M(\be; \CA; f)|\lesssim \1 f \1_2 R^{\g-\s} \:, 
\end{align}
uniformly in~$\CA$ and~$\be\in\mathbb S^{d-1}$, and  
\begin{align}\label{eq:bconesing}
	|{\sf B}(\CA; f)|\lesssim \1 f \1_2 R^{\g-\s} \:, 
\end{align} 
uniformly in~$\CA$ and $\boldsymbol\Xi$. 
\end{lem}

\begin{proof}  
Suppose now that~$d\ge 2$ and that~$\CA$ satisfies~\eqref{eq:Asing}. 
Let 
$\widehat{\boldsymbol\Xi}_\be$ be the projection of~$\boldsymbol\Xi$ 
onto the hyperplane~$\BT_\be$, and let~$\widehat{\boldsymbol\Xi}^\perp_\be$ 
be its projection onto the one-dimensional subspace spanned by~$\be$. 
Then it follows from 
%
\eqref{eq:Asing} that~$\CA(\hat\bxi, t)$ satisfies the bound 
\begin{align}\label{eq:red}
	|\p_t^k \CA(\hat\bxi; t)|
	\lesssim \lu |\hat\bxi|+ |t|\ru^{-\rho} 
(r(\hat\bxi)+ \dc(t))^{-k},	
	\quad k = 0, 1, \dots,
\end{align}
where $r = r(\hat\bxi) = 
\min\big(\dist(\hat\bxi, \widehat{\boldsymbol\Xi}_\be), 1\big)$ and 
$\dc(t) = \min\big(\dist (t, \widehat{\boldsymbol\Xi}_\be^\perp), 1\big)$. 
Using Proposition~\ref{prop:szego} with~$q=1$, and Lemma~\ref{lem:dis}, 
we get 
\begin{align}\label{eq:rhat}
	|\CM(\hat\bxi; \be; \CA; f)|\le &\ \|D_1(\CA(\hat\bxi; \,\cdot\,); \R_+ ; f)\|_{1} \notag\\
	\lesssim &\ 
	\1 f\1_2\ R^{\g - \s} \|\chi_+ \op_1(\CA(\hat\bxi; \,\cdot\,)) 
	( \mathds{1} - \chi_+)\|_{\s}^\s\notag\\
	\lesssim &\ \1 f\1_2\ R^{\g - \s}  
	\lu \hat\bxi\ru^{-\rho\s+1} \, \log \big(r(\hat\bxi)^{-1} + 2\big) \:,
	\end{align}
	which is exactly the bound~\eqref{eq:mbound}. 

Since~$\rho\s >d$ and 
\begin{align*}
r(\hat\bxi)^{-1}\lesssim 
1 + \sum_{\hat\bxi^{(j)}\in\widehat{\boldsymbol\Xi}}\frac{1}{|\hat\bxi - \hat\bxi^{(j)}|},
\end{align*}
the right-hand side of \eqref{eq:rhat} is integrable in~$\hat\bxi$ and, by definition~\eqref{eq:frakm},
\begin{align*}
	|{\mathfrak M}(\be; \CA; f)|\lesssim &\  \1 f\1_2\, R^{\g-\s}\, 
	\int_{\BT_{\be}}\, 
\lu  \hat\bxi \ru^{-\rho\s+1} \, \log \big(r(\hat\bxi)^{-1} + 2\big)\,  d\hat\bxi\\
	\lesssim &\   \1 f\1_2\, R^{\g-\s},
\end{align*}
uniformly in
$\CA$, $\boldsymbol\Xi$ and~$\be\in\mathbb S^{d-1}$. This proves~\eqref{eq:mathfr} which 
entails~\eqref{eq:bconesing} due to the definition~\eqref{eq:bcoef}.
\end{proof} 

\subsection{Trace class continuity of~$D_\a(\CA, \L; f)$}
We begin with an elementary standard 
observation:

\begin{lem}\label{lem:normcon}
Let~$B^{(\eps)}$, $\eps\in [0, 1]$ be a family of uniformly 
bounded self-adjoint operators such that~$\|B^{(\eps)}-B^{(0)}\|\to 0$ 
as~$\eps \searrow 0$. Then for any function~$g\in \plainC{}(\R)$ we have 
\begin{align*}
	\|g(B^{(\eps)}) - g(B^{(0)})\|\to 0,\quad \textup{as}\quad \eps \searrow 0 \:.
\end{align*}
\end{lem} 

\begin{proof}
Since~$\{B^{(\eps)}\}$ are uniformly bounded, we can assume that~$g$ is supported on 
some compact interval~$I$. Thus it 
would suffice to prove the required convergence for polynomials. 
By virtue of the representation
\begin{align*}
(B^{(\eps)})^l - (B^{(0)})^l = \sum_{k=0}^{l-1} 
(B^{(\eps)})^{l-1-k} \big[B^{(\eps)} - B^{(0)}\big] 
(B^{(0)})^k, \quad l = 1, 2, \dots,
\end{align*}
we have 
\begin{align*}
\|(B^{(\eps)})^l - (B^{(0)})^l \|\lesssim \|B^{(\eps)}  - B^{(0)}\| \to 0 \:,
\quad \eps \searrow 0,
\end{align*}
as required.
\end{proof}

Consider a family of \textit{continuous} symbols 
$\CA, \CA^{(\eps)}\in 
\plainC\infty(\R^d\setminus\boldsymbol\Xi; \C^{n\times n})\cap \plainC{}(\R^d; \C^{n\times n})$, 
$\eps\in (0, 1],$ 
that satisfy~\eqref{eq:Asing} with~$\rho>d$, uniformly in~$\eps$.  Suppose that  
\begin{align}\label{eq:symbolcon}
\sup_{\bxi} | \CA^{(\eps)}(\bxi) - \CA(\bxi)|\to 0,\ \eps \searrow 0 \:.
\end{align}
In the next lemma we use the notation~$D_\a(\CA) = D_\a(\CA, \L; f)$.

\begin{lem} \label{lem:cont}
Let~$d\ge 2$. Suppose that~$f$ satisfies Condition~\ref{cond:f} with some~$\g \in (0, 1]$. 
Suppose that~$\CA, \CA^{(\eps)}$ are continuous and satisfy 
\eqref{eq:Asing} with some~$\rho$ such that~$\rho \g > d$,  
uniformly in~$\eps\in (0, 1]$. Then under the condition~\eqref{eq:symbolcon} we have 
\begin{align*}
	\sup_{\a\gtrsim 1}\,\a^{1-d}\, \|D_\a (\CA^{(\eps)}) - D_\a(\CA)\|_{1}\to 0,\quad
	\eps \searrow 0 \:.
\end{align*}
\end{lem} 

\begin{proof}
Split the sum
\begin{align*}
	\|D_\a (\CA^{(\eps)}) - D_\a(\CA)\|_{1} = \sum_{k=1}^\infty\  
	s_k\big(D_\a (\CA^{(\eps)}) - D_\a(\CA)\big)
\end{align*}
into two parts:
\begin{align*}
	Z_1(\eps, \ell) = &\ \sum_{k=1}^{2\ell}\ s_k\big(D_\a (\CA^{(\eps)}) 
	- D_\a(\CA)\big),\\
	Z_2(\eps, \ell) = &\ \sum_{k = 2\ell+1}^\infty\  
	s_k\big(D_\a (\CA^{(\eps)}) - D_\a(\CA)\big),
\end{align*}
where~$\ell$ is an integer to be specified later. 
Estimate the first sum:
\begin{align}\label{eq:m1}
	Z_1(\eps, \ell)\le &\ 2\ell \| D_\a (\CA^{(\eps)}) - D_\a(\CA)\|\notag\\
	\le &\ 2\ell \bigg[\big\|f\big(W_\a(\CA^{(\eps)})\big)  
	-  f\big(W_\a(\CA)\big)\big\| + \big\| 
	W_\a(f(\CA^{(\eps)})) -  W_\a(f(\CA))\big\| \bigg]. 
\end{align}
In view of~\eqref{eq:symbolcon}, 
\[ \| \op_\a(\CA^{(\eps)}) - \op_\a(\CA)\|\to 0 \:,\quad \eps \searrow 0 \:,  \]
and hence 
\[ \| W_\a(\CA^{(\eps)}) - W_\a(\CA)\|\to 0 \]
as well, so that by Lemma~\ref{lem:normcon}, 
\[\big\|f\big(W_\a(\CA^{(\eps)})\big)  -  f\big(W_\a(\CA)\big)\big\|\to 0 \:.\]
In the same way the second term 
in~\eqref{eq:m1} also tends to zero. 
Thus for an arbitrary~$\d>0$ and sufficiently small~$\eps$ we have 
\begin{align*}
	Z_1(\eps, \ell) < 2\ell  \d.
\end{align*}
Now we estimate~$Z_2$. 
In view of the bound~\eqref{eq:kyfan} the sum~$Z_2(\eps, \ell)$ can be estimated as follows: 
\begin{align*}
	Z_2(\eps, \ell)\le 2\sum_{k=\ell}^\infty s_k\big(D_\a(\CA^{(\eps)})\big) 
	+  2\sum_{k=\ell}^\infty s_k\big(D_\a(\CA)
	\big) \:.
\end{align*}
To estimate the right-hand side we use Lemma~\ref{lem:dalpha}. Let~$q \in (1/2, 1)$ and~$\s < \min(2-q^{-1}, \g)$ 
be such that~$\rho\s q > d-1$. Then, by Lemma~\ref{lem:dalpha}, 
$D_\a(\CA^{(\eps)})\in \BS_q$ and 
\begin{align*}
	\|D_\a(\CA^{(\eps)})\|_{q}^q\lesssim \a^{d-1} \:,
\end{align*}
uniformly in~$\eps$. By~\eqref{eq:ind}, 
this allows us to write the following bound for the singular values:
\begin{align*}
	s_k(D_\a(\CA^{(\eps)}))\lesssim k^{-\frac{1}{q}}\, \a^{\frac{d-1}{q}} \:.
\end{align*}
Consequently, 
\begin{align*}
	\sum_{k=\ell}^\infty s_k(D_\a(\CA^{(\eps)}))\lesssim &\ 
	\a^{\frac{d-1}{q}}\,\sum_{k=\ell}^\infty  
	k^{-\frac{1}{q}} \lesssim  \a^{\frac{d-1}{q}}\, \ell^{1-\frac{1}{q}} \:,\\
	\textup{and hence}\quad Z_2(\eps, \ell)&\ 
	\lesssim \a^{\frac{d-1}{q}}\, \ell^{1-\frac{1}{q}} \:,
\end{align*} 
uniformly in~$\eps$. Thus, together with the bound for~$Z_1(\eps, \ell)$ we obtain that 
\begin{align*}
	\a^{1-d}\|D_\a (\CA^{(\eps)}) - D_\a(\CA)\|_{1} 
	\lesssim &\ \a^{1-d}\, \ell \d + \a^{(d-1)(\frac{1}{q}-1)}\, \ell^{1-\frac{1}{q}}\\
	=  \a^{1-d}\, \ell\d + \big(\a^{1-d} \ell\big)^{1-\frac{1}{q}} \:.
\end{align*}
Take~$\ell = M\a^{d-1}$ with some~$M>0$. Then 
\begin{align*}
	\limsup_{\eps \searrow 0}\,
	\sup_{\a\gtrsim 1}\,\a^{1-d}\|D_\a (\CA^{(\eps)}) - D_\a(\CA)\|_{1} 
	\lesssim  M \d + M^{1-\frac{1}{q}} \:.
\end{align*}
Since~$\d >0$ and~$M>0$ are arbitrary, 
the left-hand side equals zero as required. 
\end{proof}

\subsection{Approximation by smooth symbols}    
Here we assume  as before that~$\CA$ is a Hermitian matrix-valued symbol satisfying~\eqref{eq:Asing}.
To approximate~$\CA$ by a family of smooth symbols we use the following partition of unity.

For a parameter $h\in (0, 1/4]$, 
let the functions $\t^{(j)}_h\in \plainC\infty_0(\R^d)$, 
$j = 1, 2, \dots, N = \card\, \boldsymbol\Xi$, and $\z_h\in \plainC\infty(\R^d)$ be such that 
\begin{align*}
\t^{(j)}_h(\bxi) = &\ 1, \quad  \textup{for}\quad 
|\bxi-\bxi^{(j)}| \le h;\\
\t^{(j)}_h(\bxi) = &\ 0,  \quad  \textup{for}\quad |\bxi-\bxi^{(j)}| \ge 2h,
\end{align*}
and 
\begin{align*}
\sum_{j=1}^N \, \t^{(j)}_h(\bxi) + \z_h(\bxi) = 1,\quad \bxi\in\R^d.
\end{align*}
In particular, $\z_h(\bxi) = 1$ if 
$\dc(\bxi) = \dc(\bxi, \boldsymbol\Xi) := \min\{1, \dist(\bxi, \boldsymbol\Xi)\} \ge 2h$. A standard 
argument easily shows that one can construct such cut-off functions 
that for every $j = 1, 2, \dots, N,$ the following bounds hold: 
\begin{align}\label{eq:tz}
|\nabla_{\bxi}^k\, \t^{(j)}_h(\bxi)| + 
|\nabla_{\bxi}^k \, \z_h(\bxi)|
\lesssim h^{-k},\quad k = 0, 1, \dots.
\end{align}
We compare the operators~$D_\a(\CA; f)$ and~$D_\a(\CA_h; f)$, where  
\begin{align}\label{eq:ar}
\CA_h(\bxi) := \CA(\bxi) \z_h(\bxi),
\end{align}
so that 
\begin{align}\label{eq:ztil}
\CA(\bxi)-\CA_h(\bxi) = 
\CA(\bxi) \,\sum_{k=1}^N \, \t^{(j)}_h(\bxi).
\end{align}
Below we use again the shorthand notation $D_\a(\CA)$ instead of $D_\a(\CA, \L; f)$.

\begin{lem}\label{lem:compare} 
Let~$d\ge 2$. 
Suppose that~$f$ satisfies Condition~\ref{cond:f} and~$\L$ satisfies Condition~\ref{cond:domain}.
Let~$\CA$ satisfy~\eqref{eq:Asing}. 
Assume that~$h\in (0, 1]$ and~$\a h \gtrsim 1$. Then we have
\begin{align}\label{eq:compare}
\|D_\a(\CA) - D_\a(\CA_h) \|_1\lesssim (\a h)^{d-1}\,,
\end{align}
with an implicit constant uniform in~$\CA$ and the set $\boldsymbol\Xi$.
\end{lem}

\begin{proof}
First we use Proposition~\ref{prop:onfull} with 
\begin{align*}
	P  \coloneqq\chi_\L\,, \quad
	J = J_{2h} \coloneqq \op_\a (\z_{2h}), \quad 
	A = \op_\a(\CA)\,, \quad  B = \op_\a(\CA_h)\:.
\end{align*}
The norms of these operators~$P$ and~$J$ are bounded uniformly in~$h$ and,
since $\z_h\z_{2h}  = \z_{2h}$, we have~$(\CA - \CA_h) \z_{2h} = 0$.
Let us estimate the right-hand side of~\eqref{eq:onfull} by starting 
with a bound for 
\begin{align*}
	\|[(\mathds{1}-J_{2h})\op_\a(\CA_h), P]\|_{\s}^\s
\end{align*}
with an arbitrary~$\s\in (0, 1]$.  
In view of~\eqref{eq:ztil}, we have 
\begin{align*}
	(1-\z_{2h}(\bxi))\,\CA_h(\bxi)  = 
\sum_{j=1}^N\,  \CB^{(j)}_h(\bxi),\quad \textup{with}\quad 
 \CB^{(j)}_h(\bxi)
	= \CA(\bxi) \z_h(\bxi)\, \t^{(j)}_{2h}(\bxi)\:.
\end{align*}
Consider each commutator $[\op_\a(\CB^{(j)}_h), P]$ individually. 
First we rescale and shift the symbol $\CB^{(j)}_h$:
\begin{align*}
\tilde\CB^{(j)}_h(\bxi) :=  &\ \CB^{(j)}_h(\bxi^{(j)} + h \bxi)\\
	= &\ \CA(\bxi^{(j)} + h \bxi) \z_h(\bxi^{(j)} + h \bxi)\, \t^{(j)}_{2h}
	(\bxi^{(j)} + h \bxi)\:.
\end{align*}
The following unitary equivalence is easily checked:
\begin{align}\label{eq:unieq}
\op_\a(\CB^{(j)}_h) = e^{i\a \bxi^{j)}\cdot \bx} \, 
\op_{\a h}(\tilde\CB^{(j)}_h)\, e^{-i\a \bxi^{j)}\cdot \by},
\end{align}
By the definition of the cut-offs~$\t_h^{(j)}$ and~$\z_h$ we have 
\begin{align*}
|\nabla_{\bxi}^k\, \z_h(\bxi^{(j)} + h \bxi)\, 
\t^{(j)}_{2h}(\bxi^{(j)} + h \bxi)|
\lesssim 1,\quad k = 0, 1, \dots.
\end{align*}
Furthermore, since
\begin{align*}
\frac{1}{\dc(\bxi^{(j)} + h \bxi, \boldsymbol\Xi)}\lesssim 
\frac{1}{h\, \dc\big(\bxi, \widetilde{\boldsymbol\Xi}^{(j)}_h\big)},\quad 
\widetilde{\boldsymbol\Xi}^{(j)}_h = \frac{1}{h}\big(
\boldsymbol\Xi-\bxi^{(j)}\big),
\end{align*}
it follows from \eqref{eq:Asing} that
\begin{align*}
|\nabla_{\bxi}^k\, \CA(\bxi^{(j)} + h\bxi )|
\lesssim \frac{1}{\dc(\bxi, \widetilde{\boldsymbol\Xi}^{(j)}_h)^k},\quad |\bxi|<2h \:.
\end{align*}
Therefore, 
\begin{align*}
|\nabla_{\bxi}^k\, \tilde\CB^{(j)}_h\,(\bxi)|
\lesssim \frac{\chi_{2h}(\bxi)}{\dc\big(\bxi, \widetilde{\boldsymbol\Xi}^{(j)}_h\big)^k}, 
\end{align*}
where $\chi_{2h}$ is the indicator of the ball $\{\bxi: |\bxi|< 2h\}$, so $\tilde\CB^{(j)}_h$ 
satisfies \eqref{eq:Asing} with an arbitrary $\rho >0$ 
and the singular set $\widetilde{\boldsymbol\Xi}^{(j)}_h$. 
Now we can use Lemma \ref{lem:nonsm} and the relation \eqref{eq:unieq} 
to conclude that
\begin{align*}
\big\|[\op_\a( \CB^{(j)}_h), P]\big\|_{\s}^\s
	= \big\|[\op_{\a h}(\tilde \CB^{(j)}_h), \chi_\L]\big\|_{\s}^\s
	\lesssim (\a h)^{d-1} \:.
\end{align*}
Therefore, for each $\s\in (0, 1]$ we have 
\begin{align*}
	\|[(\mathds{1}-J_{2h})\op_\a(\CA_h), P]\|_{\s}^\s\lesssim (\a h)^{d-1},
\end{align*}
as required. Estimating~$\big\|[(\mathds{1}-J_{2h})\op_\a(\CA), P]\big\|_{\s}$ and 
$\|[J_{2h}, P]\|_\s = \|[\mathds{1} -J_{2h}, P]\|_\s~$, $\s \in (0, 1]$, is done in the same way.
Substituting the obtained bounds into~\eqref{eq:onfull}, we get the required estimate~\eqref{eq:compare}.
\end{proof} 

\section{Proof of Widom's formula for smooth symbols and non-smooth test functions} 
 \label{Sec:ProofSmooth}
We now prove Theorem~\ref{thm:main} for smooth symbols~$\CA$, and functions~$f$ satisfying 
Condition~\ref{cond:f2}.

\begin{thm}\label{thm:mainsmooth} 
Suppose that~$d\ge 2$. Let the function~$f$ satisfy Condition~\ref{cond:f2} 
with some~$\g\in (0, 1]$, and let~$\L$ be a bounded~$\plainC1$-region. Let 
$\CA\in\plainC\infty(\R^d; \C^{n\times n})$ be a Hermitian matrix-valued symbol such that 
\eqref{eq:wid} holds with some~$\rho$ such that~$\rho\g > d$. 
Then the formula~\eqref{eq:widom1} holds.
\end{thm}

\begin{proof}  
The proof is conducted in two steps following the idea of \cite{leschke-sobolev-spitzer}. 
In Step 1 we prove the theorem for 
$f\in \plainC2$, and in Step 2 we extend the result to the functions 
$f$ satisfying Condition~\ref{cond:f2}.

\underline{Step 1.} First suppose that~$f\in \plainC2(\R)$. 
Since the operator~$\op_\a(\CA)$ is bounded, 
we may assume that~$f\in\plainC2_0(\R)$. Thinking of~$f$ as a function 
satisfying 
Condition~\ref{cond:f} with a~$t_0$ outside its support, 
and with~$\gamma = 1$, 
we conclude that~$\1 f\1_2\lesssim \|f\|_{\plainC2}$, where~$\1 f\1_2$ is defined in~\eqref{eq:fnorm}. 
Let~$I$ be a closed interval containing the support of~$f$. 
Let~$f_\d, \d\in (0, 1]$, 
be a family of polynomials such that~$\|f-f_\d\|_{\plainC2(I)} \to 0$ as~$\d\to 0$.  
Estimate:
\begin{align*}
	\limsup_{\a\to\infty}\,\big|\a^{1-d}\,\tr D_\a(\CA; f ) - {\sf B}(\CA; f)\big|
	\le &\ \sup_{\a\ge 1}\,\a^{1-d}\,
	\big\|D_\a(\CA; f) - D_\a(\CA; f_\d)\big\|_1 \\
	&\ + \limsup_{\a\to\infty}\big|\a^{1-d}\,\tr D_\a(\CA; f_\d) - {\sf B}(\CA; f_\d)\big|\\
	&\ +\big|{\sf B}(\CA; f_\d) - {\sf B}(\CA; f)\big| \:.
\end{align*}
Since~$f_\d$ is a polynomial, we can apply 
Proposition~\ref{prop:widom} which implies that 
the second term equals zero for each~$\d >0$.  
Estimate the first and the third term using 
Lemmata~\ref{lem:dalpha} and~\ref{lem:mbound}:
\begin{gather*}
	\sup_{\a\ge 1}\,\a^{1-d}\, \big\| D_\a(\CA; f )  
	- \ D_\a(\CA; f_\d)\big\|_{1} = \sup_{\a\ge 1} \,\a^{1-d}\, \big\| D_\a(\CA; f - f_\d)\big\|_{1}
	\lesssim 
	\|f-f_\d\|_{\plainC2(I)}\to 0,\\
	\big|{\sf B}(\CA; f )  - \ {\sf B}(\CA; f_\d)\big|= | {\sf B}(\CA; f - f_\d)|
	\lesssim \|f-f_\d\|_{\plainC2(I)}\to 0 \:,
\end{gather*} 
as~$\d\to 0$. Therefore the theorem holds for~$f\in \plainC2(\R)$.

\underline{Step 2.} 
To complete the proof of the theorem we assume without loss of generality 
that~$f$ satisfies Condition~\ref{cond:f} with~$t_0=0$. 
Let~$\z\in \plainC\infty(\R)$ be such that~$\z(t) = 1$ if 
$|t|\ge 2$,~$\z(t) = 0$ if~$|t| \le 1$ and~$0\leq \z \leq 1$. Denote~$\z_r(t) = \z(t r^{-1})$, $r \in (0, 1]$. 
Split~$f$ in the sum
\begin{align*}
	f(t) = &\ g_r^{(1)}(t)+ g_r^{(2)}(t),\\
	g_r^{(1)}(t) = &\ f(t) \z_r(t),\quad  g_r^{(2)}(t) = f(t) (1-\z_r(t)) \:, 
\end{align*}
so that~$g_r^{(2)}$ is supported on the interval~$[-2r, 2r]$. Estimate:
\begin{align}
			\limsup_{\a\to\infty}\,\big|\a^{1-d}\,\tr D_\a(\CA; f ) - {\sf B}(\CA; f)\big|
		\le &\ \sup_{\a\ge 1}\,\a^{1-d}\,\big\|D_\a(\CA; f) - D_\a(\CA; g_r^{(1)})\big\|_1\notag\\
		& + \limsup_{\a\to\infty}\big|\a^{1-d}\,\tr D_\a(\CA;  g_r^{(1)}) - {\sf B}(\CA; g_r^{(1)})\big|\notag\\
		&+ \big|{\sf B}(\CA;  g_r^{(1)}) - {\sf B}(\CA; f)\big| \:. \label{eq:step2}
\end{align}
Since~$g_r^{(1)}\in \plainC2$, the second term on the right-hand side 
equals zero for each~$r >0$, as proved in Step 1.  
Calculating the derivatives   
\begin{align*}
	\big(g_r^{(2)}\big)'(t) &= f'(t)\:(1-\z_r)(t) - f(t)\:\z'(tr^{-1})\:r^{-1} \:, \\
	\big(g_r^{(2)}\big)''(t) &= f''(t)\:(1-\z_r)(t) 
	-2\:f'(t)\:\z'(tr^{-1})\:r^{-1}- f(t)\:\z''(tr^{-1})\:r^{-2} \:,
\end{align*}
we estimate
\begin{align*}
	\sup_{t \neq 0} \Big| \big(g_r^{(2)}\big)'(t)\Big||t|^{-\gamma+1} &\leq 
	\sup_{t \neq 0}\Big[ \big|f'(t)\big|\:|t|^{-\gamma+1} 
	+  \big|f(t)\big|\: |t|^{-\gamma}\: \big|\z'(tr^{-1})\big|\: \frac{|t|}{r} \Big] \:,\\
	\sup_{t \neq 0} \Big| \big(g_r^{(2)}\big)''(t)\Big||t|^{-\gamma+2} &\leq 
	\sup_{t \neq 0} \Big[ \big|f''(t)\big|\:|t|^{-\gamma+2}
	+ 2\: \big|f'(t)\big|\:|t|^{-\gamma+1}\: \big|\z'(tr^{-1})\big|\: \frac{|t|}{r}  \\
	&\phantom{\leq 
		\sup_{t \neq 0} \Big[ } + \big|f(t)\big|\:|t|^{-\gamma} \:\big|\z''(tr^{-1})\big|\: \Big(\frac{|t|}{r}\Big)^2 \Big]  \:.
\end{align*}
On the support of~$g_r^{(2)}$ we have~$\frac{|t|}{r}\leq 2$, so we conclude that
\[ 
\1 g_r^{(2)}\1_2  \lesssim \1 f \1_2\:. 
\]
Using this we can finally estimate the first and the third term on the right hand side of 
\eqref{eq:step2} using 
Lemmata~\ref{lem:dalpha} and~\ref{lem:mbound} with arbitrary~$\s\in (0, \g)$:
\begin{gather*}
	\sup_{\a\ge 1}\,
	\a^{1-d}\, \big\| D_\a(\CA; f )  - \ D_\a(\CA; g_r^{(1)})\big\|_{1}
	= \ \sup_{\a\ge 1}\,\a^{1-d}\, \| D_\a(\CA; g_r^{(2)})\|_{1}
	\lesssim 
	\1 f \1_2 \, r^{\g-\s}\to 0 \:,\\
	\big|{\sf B}(\CA; f )  - \ {\sf B}(\CA; g_r^{(1)})\big|= | {\sf B}(\CA; g_r^{(2)})|
	\lesssim 
	\1 f \1_2 \, r^{\g-\s}\to 0 \:,
\end{gather*} 
as~$r\to 0$. Therefore the right-hand side of~\eqref{eq:step2} equals zero. 

This completes the proof of the theorem. 
\end{proof}

\section{Proof of  Widom's formula for non-smooth symbols and non-smooth test functions}
\label{Sec:ProofGeneral}
In this section give the proofs of Theorems~\ref{thm:main}, \ref{thm:maineps} and of Proposition~\ref{prop:area}.

\subsection{Continuity of the asymptotic coefficient~${\sf B}(\CA; f)$}

Below the symbol~$\CA_h$ is as defined in~\eqref{eq:ar}. 

\begin{lem}\label{lem:bcont} 
Let $d\ge 2$, and let $f$ and~$\L$ be as in Theorem~\ref{thm:main}. 
Suppose that~$\CA$ satisfies \eqref{eq:Asing}. 
Then  
\begin{align*}
	{\sf B}(\CA_h; f)\to {\sf B}(\CA; f) \quad \textup{as~$h \to 0$} \:, 
\end{align*}
uniformly in~$\CA$, $\boldsymbol\Xi$  and~$f$. 
\end{lem}

\begin{proof}  
It suffices to prove the lemma for functions~$f$ satisfying Condition~\ref{cond:f}. 
By definition~\eqref{eq:ar}, $\CA_h(\bxi)=\CA(\bxi)$ for all 
$\bxi\in\R^d$ such that~$\dist(\bxi, \boldsymbol\Xi)>2h$, and hence 
$\CA_h(\hat\bxi; t)=\CA(\hat\bxi; t)$ for all 
$\hat\bxi\in\BT_{\be}$ such that~$\dist(\hat\bxi, \widehat{\boldsymbol\Xi})>2h$. Consequently,
\begin{align*}
\CM(\hat\bxi; \be; \CA_h; f) = \CM(\hat\bxi; \be; \CA; f), \quad \dist(\hat\bxi, \widehat{\boldsymbol\Xi})>2h \:.
\end{align*}
Since~$\CA_h$ satisfy the bounds~\eqref{eq:Asing} uniformly in~$h>0$, 
using~\eqref{eq:mbound} with some~$\s\in (d\rho^{-1}, \g)$, 
for the corresponding integrals~\eqref{eq:frakm} we obtain the convergence  
\begin{align*}
|{\mathfrak M}(\be; \CA_h; f) - &\ {\mathfrak M}(\be; \CA; f)|\\
\lesssim &\ 
\1 f\1_2\, 
	\int\limits_{\hat\bxi\in \BT_{\be}, \dist(\hat\bxi, \widehat{\boldsymbol\Xi})<2h}\, 
	\lu\hat\bxi\ru^{-\rho \s + 1} 
	\log \big(\dist(\hat\bxi, \widehat{\boldsymbol\Xi})^{-1} + 2\big)\,  d\hat\bxi\to 0 \:,
\end{align*}
as~$h\to 0$, uniformly in~$\CA$, $\be\in\mathbb S^{d-1}$ and functions~$f$ satisfying 
Condition~\ref{cond:f}.   
This, in turn, 
implies the uniform convergence  
${\sf B}(\CA_h; f)\to {\sf B}(\CA; f)$, as claimed. 
\end{proof}

\begin{lem}\label{lem:cont1}
Let $d\ge 2$, and let $f$ and~$\L$ be as in Theorem~\ref{thm:main}. 
Suppose that the family~$\CA^{(\eps)}$ satisfies 
Condition~\ref{cond:Aeps}. Then 
\begin{align*}
{\sf B}(\CA^{(\eps)}; f)\to {\sf B}(\CA; f),\quad \textup{as}\quad \eps \searrow 0 \:.
\end{align*}  
\end{lem}

\begin{proof} As before, assume that~$f$ satisfies Condition~\ref{cond:f}. 
For an arbitrary $h>0$ estimate:
\begin{align}\label{eq:br}
\limsup_{\eps \searrow 0}|{\sf B}(\CA^{(\eps)}; f) - {\sf B}(\CA; f)|
\le &\ \sup_{\eps\in (0, 1]}|{\sf B}(\CA^{(\eps)}; f) 
- {\sf B}(\CA^{(\eps)}_h; f)|\notag\\
+ &\ \limsup_{\eps \searrow 0}|{\sf B}(\CA^{(\eps)}_h; f) - {\sf B}(\CA_h; f)|
+  |{\sf B}(\CA_h; f) - {\sf B}(\CA; f)| \:.
\end{align} 
The symbols~$\CA^{(\eps)}_h$ and~$\CA_h$ are both smooth, and hence by Theorem~\ref{thm:mainsmooth}, for each of them the formula~\eqref{eq:widom1} holds. Furthermore, 
these symbols  satisfy~\eqref{eq:symbolcon} as~$\eps \searrow 0$. Therefore, by Lemma~\ref{lem:cont}, 
\begin{align*}
|{\sf B}(\CA^{(\eps)}_h; f) - {\sf B}(\CA_h; f)|
\le \sup_{\a\ge 1} \a^{1-d}\| D_\a(\CA^{(\eps)}_h) - D_\a(\CA_h)\|_1\to 0 \:,
\end{align*}
as~$\eps \searrow 0$. 
Thus the second term on the right-hand side of 
\eqref{eq:br} equals zero. 
By Lemma~\ref{lem:bcont}, 
the first and the third terms tend to zero as~$h\to 0$ uniformly in~$\eps$, and hence 
${\sf B}(\CA^{(\eps)}; f)\to {\sf B}(\CA; f)$, as claimed.
\end{proof}
  
 \subsection{Proof of Theorems~\ref{thm:main} and~\ref{thm:maineps}}
Since Theorem~\ref{thm:main} is a special case of Theorem~\ref{thm:maineps}, 
we concentrate on the proof of Theorem~\ref{thm:maineps}. Thus assume
that the family 
$\CA^{(\eps)}$ satisfies Condition~\ref{cond:Aeps}. 
For an arbitrary~$h\in (0, 1]$ we can write, 
as~$\a\to\infty$ and~$\eps \searrow 0$, 
\begin{align*}
\limsup &\,  \big|\a^{1-d}\,\tr\, D_\a (\CA^{(\eps)}; f) 
- {\sf B}(\CA; f)\big|\\
\le  &\ \limsup_{\a \to \infty}\,
\sup_{\eps \in [0, 1]}\,\a^{1-d}\| D_\a (\CA^{(\eps)}; f) 
- D_\a (\CA^{(\eps)}_h; f)\|_{1} \\
&\ + \limsup_{\eps \searrow 0}\,\sup_{\a\ge 1}\a^{1-d}\,\| D_\a (\CA^{(\eps)}_h; f) 
- D_\a (\CA_h; f) \|_{1} \\
&\ + \limsup_{\a\to\infty}\big|\a^{1-d} \,\tr\, D_\a (\CA_h; f) - {\sf B}(\CA_h; f)\big| \\
&\ + |{\sf B}(\CA_h; f) - {\sf B}(\CA; f)| \:.
\end{align*}
By  definition~\eqref{eq:ar} and by~\eqref{eq:Asing}, 
$\CA_h$ and~$\CA_h^{(\eps)}$ satisfy~\eqref{eq:wid} uniformly in 
$\eps\in [0, 1]$ (but not in~$h$!) and in view of Condition~\ref{cond:Aeps}, we have 
\begin{align*}
\sup_{\bxi}|\CA_h^{(\eps)}(\bxi) - \CA_h (\bxi)|\to 0,\quad\textup{as}\quad \eps \searrow 0 \:.  
\end{align*} 
Thus the second term in the bound above equals zero 
by virtue of Lemma~\ref{lem:cont}. The third term equals zero due to Theorem~\ref{thm:mainsmooth}. 
By Lemma~\ref{lem:compare}, 
the first term 
is bounded by~$h^{d-1}$ uniformly in~$\eps$, so that  
\begin{align*}
\limsup\,\big|\a^{1-d}\,\tr\, D_\a (\CA^{(\eps)}; f) 
- {\sf B}(\CA; f)\big| 
\lesssim  h^{d-1} + |{\sf B}(\CA_h; f) - {\sf B}(\CA; f)|,\ \quad \a\to\infty, \eps \searrow 0 \:.
 	\end{align*}
Since~$h \in (0, 1]$ is arbitrary, taking~$h\to 0$ and using Lemma~\ref{lem:bcont} we obtain 
\eqref{eq:double} thereby completing the proof of Theorem~\ref{thm:maineps}, and hence that 
of Theorem~\ref{thm:main}. 
\qed

\subsection{Proof of Proposition~\ref{prop:area}} 
The coefficient~$\mathfrak M(\be; \CA; f)$ is finite by~\eqref{eq:mathfr}. 
Let~$\be, \bb\in \mathbb S^{d-1}$ be two arbitrary unit vectors. 
Let~$\BR \in \mathrm{SO}(d)$ be a matrix such that  
$\BR\, \be = \bb$, 
and let~$\BQ = \BQ_{\BR}\in \mathrm{SU}(n)$ be 
such that~\eqref{eq:inva} holds. Thus
\begin{align*}
\CA(\hat\bxi; t) = \CA(\hat\bxi + t \be) = \BQ\, \CA(\BR\, \hat\bxi + t \bb) \BQ^{-1} \:,
\end{align*}
and hence, by cyclicity of trace, $\CM(\hat\bxi; \be; \CA; f) = \CM(\BR\,\hat\bxi; \bb; \CA; f)$. Integrating in~$\hat\bxi$, we get
\begin{align*}
\mathfrak M(\be; \CA; f) = &\ \frac{1}{(2\pi)^{d-1}}
\, \int_{\BT_\be} \CM(\mathbf R\, \hat\bxi; \bb; \CA; f)\, d\hat\bxi \\
 = &\ \frac{1}{(2\pi)^{d-1}}
\, \int_{\BT_\bb} \CM(\hat\bxi; \bb; \CA; f)\, d\hat\bxi = \mathfrak M(\bb; \CA; f) \:. 
 \end{align*}
Thus~$\mathfrak M(\be; \CA; f)$ is indeed $\be$-independent. Now it is clear that the formula 
\eqref{eq:bcoef} rewrites as~\eqref{eq:split}, as claimed.  
\qed

\section{Positivity of the coefficient~${\sf B}(\CA; f)$}
\label{sec:Positivity} 
The goal of this section is to investigate under which conditions on the function~$f$ 
and on the matrix-valued symbol~$\CA$ the asymptotic coefficient~${\sf B}(\CA; f)$ defined in 
\eqref{eq:bcoef} is strictly positive. 

\subsection{An abstract result}
Our starting point is the following well-known abstract fact from 
%
%
\cite{Berezin}.
Below~$\mathcal{H}$ is a 
complex separable Hilbert space, 
$P$ an orthogonal projection on~$\mathcal{H}$ and~$A$ a self-adjoint operator on~$\mathcal{H}$. 
The operator~$\CD(A, P; f)$ is defined in~\eqref{eq:cd}.

\begin{prop}
\label{PropPos}
Suppose that the spectrum of~$A$ is contained in the interval~$I \subset \R$,
and~$f: I \rightarrow \R$ is a concave function. 
 Assume that~$\CD(A, P; f)$ is 
 trace class and that~$PAP$ compact. Then~$\tr \CD(A, P; f) \geq 0$.
\end{prop}

Using this proposition we can prove the following bound in the spirit of 
\cite[Theorem~A.1]{laptev_safarov}. In the theorem below the space 
$\plainW{2,\infty}_{\tiny \rm loc}(I)$ and the notion 
of essential supremum ($\esssup$) are understood in the sense of Lebesgue measure.

\begin{thm}
\label{thm:ls}
Suppose that the spectrum of~$A$ is contained in the interval~$I \subset \R$ and that~$AP\in\BS_2$
and~$\CD(A, P; f)\in \BS_1$. 
Assume also that 
$f: I \rightarrow \R$ is a~$\plainW{2,\infty}_{\tiny \rm loc}(I)$-function such that 
\begin{align*}
\esssup_{t\in I} f''(t) = - k_0,\quad \textup{for some~$k_0 >0$}\:.
\end{align*}
Then, with the notation~$f_0(t) = - \frac{1}{2}\, t^2$, we have 
\begin{align}\label{eq:qua}
\tr \CD(A, P; f)\ge k_0 \tr \CD(A, P; f_0) = \frac{k_0}{2}\,\|(\mathds{1}-P) A P\|_2^2.
\end{align}
\end{thm}

\begin{proof}
We essentially follow the proof of \cite[Theorem~A.1]{laptev_safarov}.  
Denote~$\CD(f) = \CD(A, P; f)$ and 
\begin{align*}
g(t) = f(t) - k_0 f_0(t) \:, 
\end{align*}
so that~$\esssup_I g''(t) = 0$. Thus~$g$ is concave on~$I$, and by Proposition~\ref{PropPos}, 
\begin{align*}
\tr\,\CD(f) - k_0 \tr\,\CD(f_0) = \tr\,\CD(g) \ge 0 \:.
\end{align*}
The first trace on the left-hand side is finite by assumption, and the second one equals 
\begin{align*}   
 2\tr\,\CD(f_0) = \tr\big( P A^2 P - PAPAP\big) 
= \tr PA(\mathds{1} - P) AP = \|(\mathds{1}-P) A P\|_2^2 \:,
\end{align*} 
and hence it is also finite. This leads to the required bound~\eqref{eq:qua}.
\end{proof}

\subsection{Application to pseudo-differential operators}
Now we can apply the above results to the operator~$D_\a(\CA; f)$. 
We do not intend to consider the most general functions~$f$ satisfying Condition~\ref{cond:f2} 
with some~$\g\in (0, 1]$,  
but assume that~$f$ is real-valued, the set~$\mathbf T$ consists of two points, i.e. 
$\mathbf T = \{t_1, t_2\}$ with~$t_1 < t_2$, and that 
\begin{align*}
\esssup_{t\in (t_1, t_2)} \, f''(t) = -k_0,\ \quad \textup{where}\ k_0 >0 \:.
\end{align*}  

\begin{thm} \label{thm:lspdo}
Let~$d\ge 2$, and let~$f$ be as described above. 
Let~$\L\subset\R^d$ be a bounded region with a~$\plainC1$-boundary. 
Suppose that~$\CA$ is a non-zero Hermitian matrix-valued symbol 
that 
satisfies~\eqref{eq:Asing} with some~$\rho >0$ such that~$\rho\g > d$. 
Assume also that 
for all~$\bxi\notin \boldsymbol\Xi$ the spectrum of~$\CA(\bxi)$ belongs to the interval~$[t_1, t_2]$. 
Then~${\sf B}(\CA; f) >0$. 
\end{thm}

We precede the proof with two lemmata.

\begin{lem}\label{lem:nonvanish}
Let~$\CA\in \plainC2(\R; \C^{n\times n})$, be a Hermitian matrix-valued symbol satisfying
\begin{align}\label{eq:od}
\left|\frac{d^k}{d\xi^k} \CA(\xi)\right|\lesssim \lu\xi\ru^{-\rho} \:,
\end{align}
for some~$\rho >1$ and~$k = 0, 1, 2$. Then the operator~$\chi_{\R_-} \op_1(\CA) \chi_{\R_+}$ is 
Hilbert-Schmidt. 

If~$\CA(\xi)$ is a non-zero operator function, i.e.\ if 
$\CA$ is not identically zero matrix, then  
the Hilbert-Schmidt norm~$\|\chi_{\R_-} \op_1(\CA) \chi_{\R_+}\|_2$ 
is strictly positive.  
\end{lem}

\begin{proof}
Denote the kernel of the operator~$\op_1(\CA)$ by 
\begin{align*}
\check{\CA}(x) = \frac{1}{2\pi}\int \:e^{ix\xi}\: \CA(\xi) \:d\xi \:.
\end{align*}
Because of~\eqref{eq:od}, 
\begin{align*}
|\check\CA(x)|\lesssim \lu x\ru^{-2} \:,
\end{align*}
so that the (squared) Hilbert-Schmidt norm 
\begin{align*}
\|\chi_{\R_-} \op_1(\CA) \chi_{\R_+} \|_2^2 
= &\ \int_{-\infty}^0 \int_{0}^\infty |\check\CA(x-y)|^2 dy \:dx \\[0.3cm]
= &\ 
\sum_{k, l} \int_{-\infty}^0 \int_{0}^\infty |\check\CA_{k, l}(x-y)|^2\, dy \:dx
\end{align*}
is finite. 
Assume now that there is an interval~$I\subset\R$ such that~$\CA(\xi)\not = 0$ for all~$\xi\in I$.
Without loss of generality we may assume that the matrix entry~$\CA_{k, l}(\xi)$ with some~$k, l$, 
is not zero for all~$\xi\in I$. Since~$\CA$ is Hermitian, we also have 
\[ 
\CA_{l, k}(\xi) = \overline{\CA_{k, l}(\xi)}\not = 0, \quad \xi\in I \:. \]
As a consequence,  
$\check\CA_{k, l}(-x) = \overline{\check\CA_{l, k}(x)}$, so that the  function 
\begin{align*}
F(x) := \frac{1}{2}\, \big(|\check\CA_{k, l}(x)|^2 + |\check\CA_{l, k}(x)|^2\big)
\end{align*}
is even and not identically zero. Therefore 
there is an interval~$J\subset \R_-$ such that~$F(x) >0$ for all~$x\in J$. 
Estimating 
\begin{align*}
\|\chi_{\R_-} \op_1(\CA) \chi_{\R_+} \|_2^2 
\ge  &\  \int_{-\infty}^0 \int_{0}^\infty F(x-y)\, dy\, dx \\[0.3cm]
\ge &\  
\int_J \int_{-\infty}^x F(t)\, dt\, dx \:, 
\end{align*} 
we conclude that the Hilbert Schmidt norm on the left-hand side is strictly positive, as required. 
\end{proof}

Using the above lemma we can now show the positivity of the asymptotic coefficient 
${\sf B}(\CA; f_0)$ for the function~$f_0(t) = -t^2/2$. 

\begin{lem}\label{lem:posf}
Let~$\CA\in \plainC\infty(\R^d\setminus \boldsymbol\Xi; \C^{n\times n})$, $d\ge 2$,   
be a Hermitian operator-valued symbol satisfying~\eqref{eq:Asing} with 
$\rho > d$ that is not identically equal to the zero matrix.
Then for the function~$f_0(t) = -t^2/2$ we have 
${\sf B}(\CA; f_0) >0$. 
\end{lem}

Recall that the coefficient~$\sf B$ is finite due to Lemma~\ref{lem:mbound}.

\begin{proof} 
By definition~\eqref{eq:bcoef} it suffices to show that~$\mathfrak M(\be; \CA; f)>0$ 
for each~$\be\in \mathbb S^{d-1}$. Fix a vector~$\be$ and 
rewrite the integrand in 
\eqref{eq:frakm} 
using the notation~$A = \op_1(\CA(\hat\bxi;\,\cdot\,)), P = \chi_{\R_+}$. 
As in the proof of Theorem~\ref{thm:ls}, we obtain 
\begin{align*}
2 \CM(\hat\bxi; \be; \CA; f_0) = &\ \tr\big( P A^2 P - PAPAP\big)\\
&\ \tr PA(\mathds{1} - P) AP = \|(\mathds{1}-P) A P\|_2^2\ge 0 \:.
\end{align*}
Since~$\CA$ is a non-zero symbol, there is a ball 
$\hat B\subset \BT_\be\setminus \hat{\boldsymbol\Xi}$ 
such that for all~$\hat\bxi\in \hat B$ 
the symbol~$\CA(\hat\bxi; \,\cdot\,)$ is~$\plainC\infty(\R; \C^{n\times n})$, it is non-zero 
and satisfies
\begin{align*}
\left|\frac{d^k}{d\xi^k} \CA(\hat\bxi; \xi)\right|\lesssim \lu\xi\ru^{-\rho},\quad k = 0, 1, \dots, 
\quad\xi\in\R \:,
\end{align*}
with a constant uniform in~$\hat\bxi\in \hat B$. Thus by Lemma~\ref{lem:nonvanish}, 
$\CM(\hat\bxi; \be; \CA; f_0)>0$ for all~$\hat\bxi \in\hat B$. This leads to the positivity 
of~${\sf B}(\CA; f_0)$. 
\end{proof}

\subsection{Proof of Theorem~\ref{thm:lspdo} }
In order to use Theorem~\ref{thm:ls} we check that~$\op_\a(\CA)\chi_\L$ is Hilbert-Schmidt, i.e.\
\begin{align*}
\|\op_\a(\CA) \chi_\L\|_2^2 = \bigg(\frac{\a}{2\pi}\bigg)^d \int |\CA(\bxi)|^2\, d\bxi \ 
\int_{\L} \, d\bx < \infty \:,
\end{align*}
where we have used that~$|\CA(\bxi)|\lesssim \lu \bxi\ru^{-\rho}$ with~$\rho\ge\rho \gamma > d$.
Now, by Theorem~\ref{thm:ls},
\begin{align*}
\tr D_\a(\CA; f)\ge k_0\,\tr   D_\a(\CA; f_0),\quad f_0(t) = -\frac{1}{2}\, t^2.
\end{align*}
Using the asymptotics~\eqref{eq:widom1}  established in Theorem~\ref{thm:main}, we
obtain~${\sf B}(\CA; f)\ge k_0\, {\sf B}(\CA; f_0)$. The latter is positive by Lemma~\ref{lem:posf}. This completes the proof. 
\qed

\subsection{Corollaries for the functions~$\eta_\varkappa$} 

Let the functions~$\eta_\varkappa$ be as defined in~\eqref{eq:eta_gamma}. 
Each function~$\eta_\varkappa$ satisfies Condition~\ref{cond:f2} with~${\sf T} = \{0, 1\}$, 
with~$\gamma = \min \{\varkappa, 1\}$, if~$\varkappa\not = 1$, and 
with arbitrary~$\gamma <1$ if~$\varkappa = 1$.

\begin{cor}\label{cor:poseta}
Let~$\CA$ be as in Theorem~\ref{thm:lspdo} and such that~$0\le \CA(\bxi)\le \mathds{1}$ 
for all $\bxi\notin\boldsymbol\Xi$. If~$\varkappa\in (0, 2)$, then~${\sf B}(\CA; \eta_\varkappa) >0$.  
\end{cor}

\begin{proof}
It suffices to 
show that for~$\varkappa\in (0, 2)$ the derivative~$\esssup_{t\in (0, 1)}\eta_\varkappa''(t) <0$. 
If~$\varkappa = 1$, then one easily finds that~$\eta_1''(t) = -t^{-1}(1-t)^{-1}\le -4$. 
For~$\varkappa\not = 1$ we use a slightly modified version of the proof 
of \cite[Lemma 3.1]{LSS_2022}. 
One checks directly that 
\begin{align}\label{2nd der}
	\eta_\varkappa''(t)
	[t^\varkappa+(1-t)^\varkappa\big]^2 
	= -\varkappa[t(1-t)]^{\varkappa-2} - \frac{\varkappa}{1-\varkappa} 
	[t^{\varkappa-1} - (1-t)^{\varkappa-1}]^2\,.
\end{align}
For~$\varkappa < 1$ the right-hand side is clearly negative for 
$t\in (0, 1)$ and~$\esssup\eta_\varkappa''(t) < 0$, as required. 

It remains to consider the case~$\varkappa\in{(1, 2)}$. We rewrite~\eqref{2nd der} as   
\begin{align*}
	\eta_\varkappa''(t)[t^\varkappa+(1-t)^\varkappa]^2 
	= &\  -\frac{\varkappa}{\varkappa-1} g_{\varkappa-1}(t)\,,\\ 
	g_p(t)\ceq &\   p[t(1-t)]^{p-1} - [t^p - (1-t)^p]^2\,,
\end{align*}
for~$p:= \varkappa-1\in {(0, 1)}$. 
Since $[t^\varkappa+(1-t)^\varkappa]^2$ is strictly positive and bounded 
for~$t\in [0, 1]$, it suffices to show that~$g_p(t)\ge c$ 
with some positive~$c$. This claim is equivalent to 
\begin{align}\label{ineq3}
	[t(1-t)]^{1-p}[t^{2p} + (1-t)^{2p} + c] \le 2t(1-t) + p\,.
\end{align}
Using the notation 
\begin{align*}
	M_p\ceq 2^{p-1}\max_{t\in[0, 1]} [t^{2p} + (1-t)^{2p}] = 
	\begin{cases}
		2^{-p}& \quad\textup{if}\quad 0<p<1/2\\ 
		2^{p-1}&\quad\textup{if}\quad 1/2\le p<1
	\end{cases}\,,
\end{align*} 
the (elementary example of the) Young inequality
\begin{align*}
	ab\le \frac{a^u}{u} + \frac{b^v}{v}\,,\quad a, b \ge 0\,,\quad u, v> 1\,,\, \frac{1}{u}+\frac{1}{v} = 1
\end{align*}
for~$a = [2t(1-t)]^{1-p}\,, u = (1-p)^{-1}$ and~$b = 1, v = p^{-1}$\, yields   
 \begin{align*}
	[(t(1-t)]^{1-p} [t^{2p} + (1-t)^{2p}+c] \le &\ \big( M_p+ 2^{p-1} c\big) [2t(1-t)]^{1-p}\\
	\le &\  \big(M_p+2^{p-1} c\big)\big[(1-p)(2t(1-t)) + p\big] \,.
\end{align*}  
Since~$M_p <1$ for~$p\in (0, 1)$, the number 
\begin{align*}
c = (1-M_p) 2^{1-p}
\end{align*}
is positive. With this choice of~$c$ the right-hand side of the above inequality coincides with 
\[
(1-p)(2t(1-t)) + p \le 2t(1-t)+p\,,
\]
so~\eqref{ineq3} holds. This completes the proof of the inequality 
$\esssup_{t\in (0, 1)}\eta_\varkappa''(t) <0$ and hence entails that~${\sf B}(\CA; \eta_\varkappa)>0$. 
\end{proof}

\section{Proof of the main theorem}
\label{Sec:Appl}
We are now in a position to complete the proof of Theorem~\ref{thm:mainPhys}.
In order to use Theorems~\ref{thm:main} and~\ref{thm:maineps} 
we begin with the relation derived already in the introduction: 
 \begin{align*}
S_\varkappa(\Pi^{(\eps)}, L\L) = \tr D_\a (\CA^{(\eps)}, \L; \eta_\varkappa),
\quad \a = L\eps^{-1},
\end{align*}
 where the symbol~$\CA^{(\eps)}$ is given by 
 \begin{align*} 
\CA^{(\eps)}(\bxi) = &\ \mathcal P^{(\eps)}(\bxi \eps^{-1})
= 
\frac{1}{2}\: \Big( \mathds{1}_{\C^4} + \frac{\sum_{\beta=1}^{3}\xi_\beta \gamma^\beta  - \eps m }
{\sqrt{\bxi^2+\eps^2m^2}} \gamma^0 \Big)
\: 
\phi\big(\sqrt{\bxi^2+(\eps m)^2}\big)
%
%
\:.
\end{align*}
The symbol~$\CA^{(\eps)}$ is Hermitian $4\times 4$-matrix-valued and it 
satisfies Condition~\ref{cond:Aeps} with the limiting symbol 
\[ 
	\CA(\bxi)=\frac{1}{2}\: \Big( \mathds{1}_{\C^4} + \sum_{\beta=1}^{3}\frac{\xi_\beta}{|\bxi|}\gamma^\beta 
	\gamma^0 \Big)
	\: 
\phi(|\bxi|)	
	\]
with the finite set~${\Xi=\{\mathbf{0}\}}$ and for arbitrary~$\rho>0$. 
Moreover, as we have already observed earlier, 
each function~$\eta_\varkappa$ satisfies Condition~\ref{cond:f2} with~${\sf T} = \{0, 1\}$, 
with~$\gamma = \min \{\varkappa, 1\}$, if~$\varkappa\not = 1$, and 
with arbitrary~$\gamma <1$ if~$\varkappa = 1$. 
Thus, according to Theorem~\ref{thm:main}, as~$L \to \infty$ and~$\eps>0$ is fixed, we have 
\begin{align*}
\lim \, (L\eps^{-1})^{-2}S_\varkappa(\Pi^{(\eps)}, L \Lambda) 
= 	\lim  \a^{-2}  \, \tr\, D_{\a}(\CA^{(\eps)}, \L; \eta_\varkappa) =  
{\sf B}(\CA^{(\eps)}; \eta_\varkappa).
\end{align*}
Similarly, if~$\eps \searrow 0$ and~$\a = L\eps^{-1}\to\infty$, then 
Theorem~\ref{thm:maineps} leads to the formula 
\begin{align*}
\lim \, L^{-2} \eps^2 \,S_\varkappa(\Pi^{(\eps)}, L \Lambda)  =  {\sf B}(\CA; \eta_\varkappa) \:.   
\end{align*}
To complete the proof of~\eqref{eq:areaLaw} and~\eqref{eq:areaLawL} we will check that 
the symbols~$\CA^{(\eps)}, \CA$ satisfy the conditions of Proposition~\ref{prop:area}. 
The following lemma is the first step in this direction.  

\begin{lem}
\label{lem:RotDirac}
Let~$\BR\in \mathrm{SO}(3)$ be arbitrary. 
Then there exists a matrix~$\BQ = \BQ_{\BR}\in \mathrm{SU}(4)$ 
such that for any~$\mathbf{v} \in \R^3$:
\[ 
\BQ\: \sum_{\beta=1}^{3} (\mathbf R \mathbf v)_\beta\, \gamma^\beta \:\BQ^{-1}  
= \sum_{\beta=1}^3 v_\beta \gamma^\beta \:,\quad \textup{and} \quad 
\BQ \gamma^0 \BQ^{-1} = \gamma^0.
\]
\end{lem}
\begin{proof} This lemma is a specialization of the Lorentz invariance of the Dirac equation
to spatial rotations. General proofs can be found for example in~\cite[Section~3-4]{sakurai}
or~\cite[Lemma~1.3.1]{intro}. To make the connection clear, we recall that for any
orthochronous proper Lorentz transformation $\Lambda \in \mathrm{SO}(1,3)$ there is
matrix~$U \in \mathrm{SU}(2,2)$ which is unitary with respect to the spin inner product~\eqref{sip}
such that~$U \Lambda^i_j \gamma^j U^{-1} = \gamma^i$ for all~$i \in \{0,\ldots3\}$.
Specializing this result to the case~$\Lambda=\BR$ of rotations, the resulting matrix~$U=: \BQ$ commutes
with~$\gamma^0$, implying that it is also in~${\mathrm{SU}}(4)$, concluding the proof.

For the reader who is not so familiar with Dirac spinors, we now give an alternative, more computational
proof. Using the method of Euler angles, it suffices to prove this lemma for rotations around the three coordinate axes, because any other rotation may be written as a product of those three rotations.
Let us consider for example the case where ~$\BR$ is a rotation around the $z$-axis
(for rotations around the $x$- and~$y$-axes, the computation is similar).
Then~${\mathbf{R}}$ is given by
\[ \mathbf R= \begin{pmatrix}
	\cos \theta & -\sin \theta &0  \\
	\sin \theta & \cos \theta  &0 \\
	0 & 0 & 1
\end{pmatrix}\:, \]
where~$ \theta \in \R$ is the rotation angle.
We claim that 
\[  \BQ:= \begin{pmatrix}
	e^{-i\theta/2} & & & \\
	& e^{i\theta/2} &  &  \\
	&  & e^{-i\theta/2} &  \\
	&  &  & e^{i\theta/2}
\end{pmatrix} \:
\]
is the sought matrix.
Indeed, note that
\begin{flalign*}
\BQ\:  \gamma^0\: \BQ^{-1}  &=\gamma^0  \:, 
\quad \BQ\:  \gamma^3\: \BQ^{-1}   =\gamma^3  \:, \\[0.2cm]
\BQ\: \gamma^1 \:\BQ^{-1} &= \begin{pmatrix}
	& & &  e^{i\theta} \\
	& & e^{-i\theta}  & \\
	& -e^{i\theta} & & \\
	-e^{-i\theta} & & & 
\end{pmatrix}\:, \\[0.2cm]
\BQ\: \gamma^2\: \BQ^{-1} &= \begin{pmatrix}
	& & &  -ie^{i\theta} \\
	& & ie^{-i\theta}  & \\
	& ie^{i\theta} & & \\
	-ie^{-i\theta} & & & 
\end{pmatrix} \:.
\end{flalign*}
Then, by a straightforward computation we see that
\begin{flalign*}
&\BQ\:\sum_{\beta=1}^{3} (\mathbf R \mathbf v)_\beta \: \gamma^\beta\: \BQ^{-1 }\\
&= v_1 \Big( \cos\theta \:\BQ \:\gamma^1\: \BQ^{-1} 
+ \sin\theta \:\BQ \:\gamma^2\: \BQ^{-1} \Big)
 + v_2 \Big( -\sin\theta \:\BQ\:\gamma^1\: \BQ^{-1} 
 + \cos\theta \: \BQ\:\gamma^2\: \BQ^{-1} \Big) + v_3 \gamma^3 \\
&= \sum_{\beta=1}^{3} v_\beta \gamma^\beta \:,
\end{flalign*}
which concludes the proof.
\end{proof}

%
Lemma~\ref{lem:RotDirac} ensures that 
the symbols~$\CA$ and~$\CA^{(\eps)}$ satisfy the conditions of Proposition~\ref{prop:area}. 
Therefore,
\begin{align*}
{\sf B}(\CA^{(\eps)}; \eta_\varkappa) = \mathfrak M_\varkappa^{(\eps)}\, 
\mathrm{vol}_2(\partial \Lambda) ,\quad  \,
{\sf B}(\CA; \eta_\varkappa) =  \mathfrak M_\varkappa\, \mathrm{vol}_2(\partial \Lambda) \:,
\end{align*}
where
\begin{align*}
\mathfrak M_\varkappa^{(\eps)} :=  \mathfrak M(\be; \CA^{(\eps)}; \eta_\varkappa),
\quad \mathfrak M_\varkappa :=  \mathfrak M(\be; \CA; \eta_\varkappa) \:,
\end{align*}
with an arbitrary unit vector~$\be$, and the coefficients $\mathfrak{M}$ are defined in \eqref{eq:frakm}. 
By continuity of the asymptotic coefficient established in Lemma~\ref{lem:cont1}, we have 
$\mathfrak M_\varkappa^{(\eps)} \to \mathfrak M_\varkappa$ as~$\eps \searrow 0$.

Finally, 
since the matrix symbols~$\CA$ and~$\CA^{(\eps)}$ satisfy the bounds~$0\le \CA(\bxi)\le \mathds{1}$ 
and~$0\le \CA^{(\eps)}(\bxi)\le \mathds{1}$ for 
all~$\bxi\not = \bold{0}$, it follows from 
Corollary~\ref{cor:poseta} that~${\sf B}(\CA; \eta_\varkappa) > 0$ 
and~${\sf B}(\CA^{(\eps)}; \eta_\varkappa)>0$ for~$\varkappa\in (0, 2)$. 
This immediately implies that 
$\mathfrak M_\varkappa^{(\eps)}>0$ and~$\mathfrak M_\varkappa > 0$ for 
$\varkappa\in (0, 2)$, as claimed. 
The proof of Theorem~\ref{thm:mainPhys} is now complete. 

\Thanks{\em{Acknowledgments:}}
M.L.\ gratefully acknowledges support by the Studienstiftung des deutschen Volkes and the Marianne-Plehn-Programm. The authors are grateful to the referees for their instructive remarks.


\providecommand{\bysame}{\leavevmode\hbox to3em{\hrulefill}\thinspace}
\providecommand{\MR}{\relax\ifhmode\unskip\space\fi MR }
\providecommand{\MRhref}[2]{%
  \href{http://www.ams.org/mathscinet-getitem?mr=#1}{#2}
}
\providecommand{\href}[2]{#2}

\end{document}